\documentclass{article}

\usepackage{lineno,hyperref}
\modulolinenumbers[5]
\usepackage{amssymb}
\usepackage{amsfonts}
\usepackage{amsmath}
\usepackage{lipsum}
\usepackage{amsbsy}
\usepackage{graphicx}
\usepackage{enumerate}
\usepackage{geometry}
\usepackage{tikz}
\usepackage{proof}
\providecommand{\U}[1]{\protect\rule{.1in}{.1in}}

\allowdisplaybreaks
\newtheorem{theorem}{Theorem}

\newtheorem{definition}[theorem]{Definition}
\newtheorem{example}[theorem]{Example}

\newtheorem{lemma}[theorem]{Lemma}

\newtheorem{proposition}[theorem]{Proposition}

\newenvironment{proof}[1][Proof]{\noindent\textbf{#1.} }{\ \rule{0.5em}{0.5em}}
\newcommand{\hide}[1]{}
\newcommand{\cupdot}{\mathbin{\mathaccent\cdot\cup}}
\newcommand{\bigcupdot}{\bigcup\mkern-12.5mu\cdot\mkern6mu}
\newcommand{\avg}{\rm{avg}}
\newcommand{\maj}{\rm{maj}}
\newcommand{\val}{\rm{val}}

\begin{document}

		\title{Weighted PCL over product valuation monoids}
		\author{Vagia Karyoti$^{a}$ and Paulina Paraponiari$^{b,}$\thanks{\protect\includegraphics[height=0.4cm]{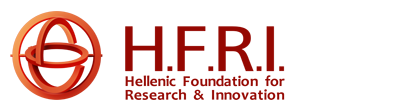}The research work was supported by the Hellenic Foundation for Research and Innovation (HFRI) under the HFRI PhD Fellowship grant (Fellowship Number: 1200).}	 \\Department of Mathematics\\Aristotle University of Thessaloniki\\54124 Thessaloniki, Greece\\ $^{a}$vagiakaryo@math.auth.gr, $^{b}$parapavl@math.auth.gr }
		
		\date{ }
		\maketitle

		\begin{abstract}
				We introduce a weighted propositional configuration logic over a product valuation monoid. Our logic is intended to serve as a specification language for software architectures with quantitative features such as the average of all interactions' costs of the architecture and the maximum cost among all costs occurring most frequently within a specific number of components in an architecture. We provide formulas of our logic which describe well-known architectures equipped with quantitative characteristics. Moreover, we prove an efficient construction of a full normal form which leads to decidability of equivalence of formulas in this logic. 
		\end{abstract}

	\section{Introduction}
	
		Architectures are a critical issue in design and development of  complex software systems since they characterize coordination principles among the components of a system. Whenever the construction of a software system is based on a ``good'' architecture, then the system satisfies most of its functional and quality requirements. Well-defined architectures require a formal treatment in order to efficiently characterize their properties. A recent work towards this direction is \cite{Ma:Co}, where the authors introduced propositional configuration logic (PCL for short) which was proved sufficient enough to describe architectures: the meaning of every PCL formula is a configuration set, which intuitively represents permissible component connections, and every architecture can be represented by a configuration set on the collection of its components. Furthermore, the authors of \cite{Ma:Co} studied the relation among architectures and architecture styles, i.e., architectures with the same types of components and topologies.

	PCL is a specification logic of software architectures which is able to describe their qualitative features. However, several practical applications require also quantitative characteristics of architectures such as the cost of the interactions among the components of an architecture, the time needed, or the probability of the implementation of a concrete interaction.   For instance, several IoT and cloud applications, which are based on Publish/Subscribe architecture, require quantitative features \cite{Ol:AP,Pa:Pu,Ya:Pr}. Moreover, considering a set of components and an architecture style, there may occur several architectures where each of them has a specific amount of some resource (e.g. memory or energy consumption). In such a setting,  the most suitable architecture must be chosen, depending on the available resources or the performance. Generally, quantitative properties are essential for performance related properties and for resource-constrained systems.
	
	The authors in \cite{Pa:On} introduced and investigated a weighted PCL (wPCL for short) over a commutative semiring $(K, \oplus, \otimes ,0,1)$ which serves as a specification language for the study of software architectures with quantitative features such as the maximum cost of an architecture or the maximum priority of a component. Nevertheless, operations like average for response time or power consumption cannot be described within the algebraic structure of semirings. Such operations are important for practical applications and have been investigated for weighted automata in \cite{Quant_Lang,expres_henz,comp_avera}. In \cite{Dr:Reg,Dr:Av} the authors provided valuation monoids as a general algebraic framework, which describe several operations that cannot fit in the structure of semirings. More recently, in \cite{proba_nested} nested weighted automata have been considered under probabilistic semantics for expressing properties such as ``the long-run average resource consumption is below a threshold". Also, the authors in \cite{comp_avera} presented algorithms which are designed specifically for computing the average response time on graphs, game graphs, and Markov chains. 
	
	However, the aforementioned works have not been developed  for the setting of systems' architectures and therefore cannot express characteristics such as the average cost of an architecture or the maximum most frequent cost/priority that occurs in an architecture. In this paper, we tackle this problem by extending the work of \cite{Pa:On}. Specifically, we introduce and investigate a weighted PCL over product valuation monoids (w$_{\text{pvm}}$PCL for short) which is proved sufficient to serve as a specification language for software architectures with important quantitative features that are not covered in \cite{Pa:On}.
	
	The contributions of our work are the following. We introduce the syntax and semantics of w$_{\text{pvm}}$PCL. The semantics of w$_{\text{pvm}}$PCL formulas are polynomials with values in the product valuation monoid. Then, in our main result, we prove that for every w$_{\text{pvm}}$PCL formula we can effectively construct an equivalent one in full normal form, which is unique up to the equivalence relation. The second main result is the decidability of equivalence of w$_{\text{pvm}}$PCL formulas. Lastly, we describe in a strict logical way several well-known software architectures with quantitative characteristics.

	\section{Preliminaries}
	
	In this section, we recall valuation monoids and product valuation monoids \cite{Dr:Av}. A \emph{valuation monoid} $(D, \oplus, {\val}, 0)$ consists of a commutative monoid $(D, \oplus, 0)$ and a valuation function ${\val}: D^+ \rightarrow D$, where $D^+$ denotes the set of nonempty finite words over $D$,  with ${\val}(d) = d$ for all $d\in D$ and ${\val}(d_1, \dots, d_n) = 0$ whenever $d_i=0$ for some $i\in \{1,\dots ,n\}.$ 
	
	$(D, \oplus, {\val}, \otimes, 0, 1)$ is a \emph{product valuation monoid}, or \emph{pv-monoid} for short if $(D, \oplus, {\val}, 0)$ is a valuation monoid, $\otimes: D^2 \rightarrow D$ is a binary operation, $1\in D$ with ${\val}(1)_{1\leq i \leq n} = 1$ for all $n\geq 1$ and $0\otimes d = d \otimes 0 = 0$, $1\otimes d = d\otimes 1=d $ for all $d\in D$. The pv-monoid is denoted simply by $D$ if the operations and the constant elements are understood. A pv-monoid $D$ is \emph{left-$\oplus$-distributive} if $d\otimes (d_1 \oplus d_2) = (d\otimes d_1) \oplus (d\otimes d_2)$ for any $d,d_1, d_2 \in D.$ \emph{Right-$\oplus$-distributivity} is defined analogously. If a pv-monoid $D$ is both left- and right-$\oplus$-distributive, then it is \emph{$\oplus$-distributive}. If $\otimes$ is associative, then $D$ is called \emph{associative}. \hide{A pv-monoid D is called \emph{left-multiplicative} if $d\otimes {\val}(d_1, d_2, \dots, d_n) ={\val}(d\otimes d_1, d_2, \dots, d_n)$ for every $n\geq 1 $ and $d,d_i \in D $ with $i\in \{1, \dots, n\}.$}We call $D$ \emph{left-{\val}-distributive} if for all $n\geq 1$ and $d, d_i\in D$ with $i\in \{1, \dots, n\}$, it holds that $d\otimes {\val}(d_1, \dots, d_n) = {\val}(d\otimes d_1, \dots, d\otimes d_n).$ \hide{Moreover, $D$ is \emph{conditionally commutative} if for any $ n\geq1$ and any two sequences $(d_1,...,d_n) , (d_1^\prime,...,d_n^\prime)$ from $ D $ with $d_i \otimes d_j^\prime = d_j^\prime \otimes d_i $ for all $ 1\leq j <i \leq n $, we have 
	 ${\val}(d_1,...,d_n) \otimes {\val}(d_1^\prime,...,d_n^\prime)={\val}(d_1 \otimes d_1^\prime, ...,d_n \otimes d_n^\prime)$.}Moreover, the pv-monoid $D$ is called (additively) idempotent if $d\oplus d =d$ for every $d\in D$. 
 
 In the following we recall some pv-monoids from \cite{Dr:Av}. The algebraic structures $(\mathbb{R} \cup\{-\infty\}, \max, {\avg},$ $ +, -\infty, 0)$ and $(\mathbb{R}\cup\{+\infty\},$ $ \min, {\avg}, +, +\infty, 0)$ with ${\avg}(d_1, \dots, d_n) = \frac{1}{n} \sum_{i=1}^n d_i$ are pv-monoids. More precisely, they are $\oplus$-distributive and left-${\val}$-distributive pv-monoids. Also, the structure $(\mathbb{R}\cup \{ -\infty,+\infty \},$ $ \min, {\maj}, \max, +\infty, -\infty)$, where ${\maj}(d_1, \dots, d_n)$ is the greatest value among all values that occur most frequently among $d_1, \dots, d_n$, is a $\oplus$-distributive pv-monoid but not left-${\val}$-distributive. Both ${\avg}$ and ${\maj}$ are symmetric functions, i.e., the value of the function given $n$ arguments is the same no matter the order of the arguments. Moreover, the pv-monoids mentioned before are idempotent.
 \begin{quote}
 \emph{Throughout the paper $(D, \oplus, {\val}, \otimes, 0,1)$ will denote an idempotent pv-monoid where ${\val}$ is symmetric. }
 \end{quote}

	Let $Q$ be a set. A \emph{formal series} (or simply \emph{series}) \emph{over}
	$Q$ \emph{and} $D$ is a mapping $s:Q\rightarrow D$. The \emph{support of} $s$ is the set $\mathrm{supp}(s)=\{q \in Q \mid s(q) \neq 0  \}$. A series with finite support is called also a \emph{polynomial}. We denote by $D\left\langle  Q
	\right\rangle $ the class of all polynomials over $Q$ and $D$.

\section{Weighted propositional interaction logic}
  In this section, we introduce the weighted propositional interaction logic over pv-monoids. Firstly, we recall from \cite{Ma:Co} the propositional interaction logic.
  
  Let $P$ be a nonempty finite set of \emph{ports}. We let $I(P)=\mathcal{P}(P) \backslash \{ \emptyset \}$, where $\mathcal{P}(P)$ denotes the power set of $P$. Every set $\alpha \in I(P)$ is called an \emph{interaction.} The syntax of \textit{propositional interaction logic} (PIL for short) formulas over $P$ is given by the grammar 
  \[ \phi :: = true \mid p \mid \overline{\phi} \mid \phi \vee \phi \]
  where $p\in P$. As usual, we set $\overline{\overline{\phi}} = \phi$ for every PIL formula $\phi$ and $false=\overline{true}.$ Hence, the conjunction of two PIL formulas $\phi, \phi^\prime$ is defined by $\phi \wedge \phi^\prime = \overline{\left(\overline{\phi} \vee \overline{\phi^\prime}\right)}$. A PIL formula of the form $p_1\wedge \dots \wedge p_n$ with $n>0$, and $p_i\in P$ or $p_i=\overline{p_i^\prime}$ with $p_i^\prime \in P$ for every $1\leq i \leq n$, is called a \textit{monomial}. For simplicity we denote a monomial $p_1\wedge \dots \wedge p_n$ by $p_1\dots p_n.$ Monomials of the form $\bigwedge_{p\in P_+} p \wedge \bigwedge_{p\in P_{-}} \overline{p}$ with $P_+ \cup P_{-} = P$ and $P_+\cap P_{-} = \emptyset$ are called \textit{full monomials.}

  Let $\phi$ be a PIL formula and $\alpha$ an interaction. We define the satisfaction relation $\alpha \models_i \phi$ by induction on the structure of $\phi$ as follows:
  
  \begin{tabular}{l l}
  	- $\alpha \models_i true,$ & \hspace*{2cm}- $\alpha \models_i \overline{\phi} $ \hspace*{0.4cm}iff \hspace*{0.1cm} $\alpha \not \models_i \phi$, \\ - $\alpha \models_i p $\hspace*{0.4cm} iff\hspace*{0.1cm} $p\in \alpha,$ & \hspace*{2cm}- $\alpha \models_i \phi_1 \vee \phi_2 $\hspace*{0.3cm} iff \hspace*{0.1cm}$\alpha \models_i \phi_1$ or $\alpha \models_i \phi_2$.
  \end{tabular}
\medskip
  
 For every $\alpha \in I(P)$ it holds $\alpha \not \models_i false.$ Moreover, for every interaction $\alpha\in I(P)$ we define its characteristic monomial $m_{\alpha} = \bigwedge_{p\in \alpha} p \wedge \bigwedge_{p\not \in \alpha} \overline{p}.$ A characteristic monomial $m_{\alpha}$ is actually a full monomial that formalises the interaction $\alpha$. Then, for every $\alpha^\prime \in I(P)$ we trivially get $\alpha^\prime \models_i m_{\alpha}$ iff $\alpha^\prime = \alpha.$
 
 \begin{quote}
 	\emph{Throughout the paper $P$ will denote a nonempty finite set of ports.} 
 \end{quote}
 
  \begin{definition}
  	Let $D$ be a pv-monoid. Then, the syntax of formulas of weighted \emph{PIL} (\emph{w$_{\text{pvm}}$PIL} for short) over $P$ and $D$ is given by the grammar
  	\[ \varphi ::= d \mid \phi \mid \varphi \oplus \varphi  \mid \varphi \otimes \varphi   \]
  	where $d\in D$ and $\phi$ denotes a \emph{PIL} formula over P.
  \end{definition}
  
  We denote by $PIL(D,P)$ the set of all w$_{\text{pvm}}$PIL formulas over $P$ and $D$. Next, we present the semantics of formulas $\varphi \in PIL(D,P)$ as polynomials $\left\Vert \varphi \right\Vert \in  D
  \left\langle I(P) \right\rangle $. For the semantics of PIL formulas $\phi$ over $P$ we use the satisfaction relation as defined above. Hence, the semantics of PIL formulas $\phi$ gets only the values $0$ and $1$.
  
 \begin{definition}
 	Let $\varphi \in PIL(D,P)$. The semantics of $\varphi$ is a polynomial  $\left\Vert \varphi\right\Vert \in D
 	\left\langle I(P) \right\rangle $. For every $\alpha \in I(P)$ the value $\left\Vert \varphi \right\Vert (\alpha)$ is defined inductively on the structure of $\varphi$ as follows:
 	
 	\begin{tabular}{l l}
 		- $\left\Vert d\right\Vert (\alpha) = d,$ & \hspace*{1cm}- $\left\Vert \varphi_1 \oplus \varphi_2 \right\Vert (\alpha) = \left\Vert \varphi_1\right\Vert (\alpha) \oplus \left\Vert \varphi_2\right\Vert (\alpha) ,$ \\ - $\left\Vert \phi \right\Vert (\alpha)  =  \left\{ \begin{array}{l l}
 		1 & \text{if } \alpha\models_i \phi \\
 		0 & \text{otherwise}
 		\end{array}  \right. ,  $ & \hspace*{1cm}- $\left\Vert \varphi_1 \otimes \varphi_2 \right\Vert (\alpha) = \left\Vert \varphi_1\right\Vert (\alpha) \otimes \left\Vert \varphi_2\right\Vert (\alpha) $.
 	\end{tabular}

 \end{definition}

\section{Weighted propositional configuration logic}
In this section, we introduce and investigate the weighted propositional configuration logic over pv-monoids. But first, we recall the propositional configuration logic (PCL for short) from \cite{Ma:Co}. The syntax of PCL formulas over $P$ is given by the grammar
\[ f::= true \mid \phi \mid \neg f \mid f\sqcup f \mid f+f \] 
where $\phi$ denotes a PIL formula over $P.$ The operators $\neg$, $\sqcup$, and $+$ are called \textit{complementation}, \textit{union}, and \textit{coalescing}, respectively. The \textit{intersection} $\sqcap$ and \textit{implication} $\implies$ operators are defined, respectively, as follows:

\begin{tabular}{l l}
	-  $f_1 \sqcap f_2 := \neg  \left( \neg f_1 \sqcup \neg f_2 \right) $, & \hspace*{1cm}-  $f_1 \implies f_2 := \neg f_1 \sqcup f_2 $.
\end{tabular}
\medskip

We let $C(P)= \mathcal{P}(I(P)) \backslash \{ \emptyset \} $. For every PCL formula $f$ and $\gamma \in C(P)$ the satisfaction relation $\gamma\models f$ is defined inductively on the structure of $f$ as follows:

\begin{tabular}{l l}
	-  $\gamma \models true,$ & \\ 
	-  $ \gamma \models \phi$ & iff $\alpha\models_i \phi$ for every $\alpha \in \gamma,$ \\ 
	-  $\gamma \models \neg f$ & iff $\gamma \not \models f,$ \\ 
	-  $\gamma \models f_1 \sqcup f_2$ & iff $\gamma\models f_1$ or $\gamma\models f_2$, \\
	- $\gamma \models f_1 + f_2$ & iff  there exist $\gamma_1, \gamma_2 \in C(P)$ such that $\gamma=\gamma_1 \cup \gamma_2, $ and $\gamma_1 \models f_1$ and $\gamma_2 \models f_2$.
\end{tabular}

\medskip

\noindent We define the \textit{closure} $\sim f$ of every PCL formula $f$ by
\begin{itemize}
	\item[-] $\sim f := f+true.$ 
\end{itemize}

\noindent Two PCL formulas $f, f^\prime$ are called \textit{equivalent}, and we denote it by $f\equiv f^\prime $, whenever $\gamma \models f$ iff $\gamma \models f^\prime$ for every $\gamma \in C(P)$. We refer the reader to \cite{Ma:Co} and \cite{Pa:On} for properties of PCL formulas.

Next, we introduce our weighted PCL over pv-monoids.

	\begin{definition}
		Let $D$ be a pv-monoid. Then, the syntax of formulas of the weighted \emph{PCL} (\emph{w$_{\text{pvm}}$PCL} for short) over P and D is given by the grammar 
		\[
		\zeta ::= d\mid f \mid \zeta \oplus \zeta \mid \zeta \otimes \zeta \mid \zeta \uplus \zeta \mid \ast \zeta 
		\]
		where $d\in D$, $f$ denotes a \emph{PCL} formula over $P$, and $\uplus$ denotes the coalescing operator among \emph{w$_{\text{pvm}}$PCL} formulas. The operator $\ast$ is called valuation operator. 
	\end{definition}

 We denote by \emph{PCL(D,P)} the set of all w$_{\text{pvm}}$PCL formulas over \emph{P} and \emph{D}. We present the semantics of formulas $\zeta \in PCL(D,P)$ as polynomials $ \left\Vert \zeta \right\Vert \in D\left\langle C(P)\right\rangle$. For the semantics of PCL formulas we use the satisfaction relation as defined previously. 

\begin{definition}
	Let $\zeta \in PCL(D,P)$. The semantics of $\zeta$ is a polynomial $\left\Vert \zeta \right\Vert \in D\left\langle C(P)\right\rangle$ where for every $\gamma \in C(P)$ the value $\left\Vert \zeta \right\Vert (\gamma) $ is defined inductively as follows:
	\begin{itemize}
		\item[-] $\left\Vert d \right\Vert (\gamma) = d$,
		\item[-] $\left\Vert f\right\Vert (\gamma)=\left\{
		\begin{array}
		[c]{rl}%
		1 & \textnormal{ if }\gamma\models f\\
		0 & \textnormal{ otherwise}%
		\end{array}
		,\right.  $
		\item[-] $\left\Vert \zeta_1 \oplus \zeta_2\right\Vert (\gamma) = \left\Vert \zeta_1\right\Vert (\gamma) \oplus \left\Vert \zeta_2\right\Vert(\gamma)  $,
		\item[-] $\left\Vert \zeta_1 \otimes \zeta_2\right\Vert (\gamma) = \left\Vert \zeta_1\right\Vert (\gamma) \otimes \left\Vert \zeta_2\right\Vert(\gamma)  $,
		\item[-] $\left\Vert \zeta_1 \uplus \zeta_2\right\Vert(\gamma) = \bigoplus_{\gamma_1\cupdot \gamma_2 = \gamma} \left( \left\Vert\zeta_1\right\Vert(\gamma_1) \otimes \left\Vert\zeta_2\right\Vert(\gamma_2) \right)$, 
		\item[-] $  \left\Vert \ast \zeta \right\Vert (\gamma) = \bigoplus_{n>0} \bigoplus_{\bigcupdot_{i=1}^{n}\gamma_i =\gamma   }{\val} \left( \left\Vert\zeta \right\Vert (\gamma_1) , \dots, \left\Vert\zeta \right\Vert (\gamma_n)  \right)   $
		 \end{itemize}
	 where $\cupdot$ denotes the pairwise disjoint union of the sets $\gamma_1, \dots, \gamma_n$ for every $n>0.$ \hide{Let $n>0$ and the sets $\gamma_i, \gamma_i^\prime \in C(P)$ where $i\in \{1, \dots, n\}$ such that $\bigcupdot_{i=1}^n \gamma_i = \bigcupdot_{i=1}^n \gamma_i^\prime = \gamma$. Then, since ${\val}$ is a symmetric function and $D$ is idempotent, there exists $i\in \{1, \dots, n\}$ such that $\gamma_i \not = \gamma_j^\prime$ for every $j\in \{1, \dots, n\}$.}
\end{definition}

It is important to note here that since the semantics of every w$_{\text{pvm}}$PCL formula is defined on $C(P)$, the sets $\gamma_1$ and $\gamma_2$ in $\left\Vert \zeta_1 \uplus \zeta_2\right\Vert (\gamma) $ and the sets $\gamma_1, \dots, \gamma_n$ in $\left\Vert \ast \zeta \right\Vert (\gamma)$  are nonempty. Trivially in $\left\Vert \ast \zeta \right\Vert(\gamma) $, the maximum value of $n$ is $|\gamma|$, i.e., the cardinality of $\gamma$. Hence, $\left\Vert \ast \zeta \right\Vert (\gamma) = \bigoplus_{n\in \{1, \dots, |\gamma| \}} \bigoplus_{\bigcupdot_{i=1}^{n}\gamma_i =\gamma   }{\val} \left( \left\Vert\zeta \right\Vert (\gamma_1) , \dots, \left\Vert\zeta \right\Vert (\gamma_n)  \right)   $. Moreover, in $\left\Vert\ast \zeta \right\Vert (\gamma)$, let the sets $\gamma_i \in C(P)$ where $i\in \{1, \dots,n\}$ and $\bigcupdot_{i=1}^n \gamma_i = \gamma$. Consider $(i_1, \dots, i_n)$ be a permutation of $(1, \dots,n)$. Then ${\val} (\left\Vert\zeta \right\Vert(\gamma_1), \dots, \left\Vert\zeta \right\Vert(\gamma_n) ) = {\val} (\left\Vert\zeta \right\Vert(\gamma_{i_1}), \dots, \left\Vert\zeta \right\Vert(\gamma_{i_n}) ) $. Hence, ${\val} (\left\Vert\zeta \right\Vert(\gamma_1), \dots, \left\Vert\zeta \right\Vert(\gamma_n) ) \oplus $ \newline ${\val} (\left\Vert\zeta \right\Vert(\gamma_{i_1}), \dots, \left\Vert\zeta \right\Vert(\gamma_{i_n}) )   = $ $ {\val} (\left\Vert\zeta \right\Vert(\gamma_1), $ $\dots, \left\Vert\zeta \right\Vert(\gamma_n) )$ by the idempotency of $D$. Therefore, for every analysis of $\gamma=\bigcupdot_{i=1}^n \gamma_i$, the value ${\val} \left( \left\Vert \zeta \right\Vert(\gamma_1), \dots, \left\Vert \zeta \right\Vert(\gamma_n) \right)$ in $\left\Vert \ast \zeta \right\Vert(\gamma) $ is computed only once. 

Two w$_{\text{pvm}}$PCL formulas $\zeta_1, \zeta_2$ are called equivalent, and we write $\zeta_1 \equiv \zeta_2$, whenever $\left\Vert \zeta_1 \right\Vert (\gamma)= \left\Vert \zeta_2 \right\Vert(\gamma) $ for every $\gamma\in C(P)$. The \textit{closure} $\sim \zeta$ of every w$_{\text{pvm}}$PCL formula $\zeta \in PCL(D,P)$ is determined by:
\begin{itemize}
	\item[-] $ \sim \zeta := \zeta \oplus (\zeta \uplus 1).$
\end{itemize}

\begin{lemma}\label{clos_form} Let $\zeta \in PCL(D,P)$. Then 
	$$\left\Vert {\sim}\zeta\right\Vert (\gamma)=\bigoplus
	\nolimits_{\gamma^{\prime}\subseteq\gamma}\left\Vert \zeta\right\Vert
	(\gamma^{\prime})$$
	for every $\gamma \in C(P)$.
\end{lemma}
\begin{proof}
	We compute
	\begin{align*}
	\left\Vert{\sim}\zeta\right\Vert (\gamma) &= \left\Vert \zeta \oplus (\zeta\uplus 1)\right\Vert (\gamma) 
	=\left\Vert \zeta  \right\Vert(\gamma) \oplus \left\Vert\zeta\uplus 1 \right\Vert (\gamma) \\
	& = \left\Vert \zeta  \right\Vert (\gamma)\oplus \left( \bigoplus\nolimits_{\gamma=\gamma'\cupdot\gamma''}\left(  \left\Vert
	\zeta\right\Vert (\gamma')\otimes\left\Vert 1\right\Vert
	(\gamma'')\right)  \right) \\
	& = \left\Vert \zeta  \right\Vert (\gamma)\oplus  \bigoplus\nolimits_{\gamma' \varsubsetneq \gamma} \left\Vert
	\zeta\right\Vert (\gamma') \\
	&= \bigoplus\nolimits_{\gamma' \subseteq \gamma} \left\Vert
	\zeta\right\Vert (\gamma')
	\end{align*}
	for every $\gamma \in C(P)$, where the fourth equality holds since $\gamma'$ and $\gamma''$ are disjoint.
\end{proof}

\medskip

Next, we present several properties of our w$_{\text{pvm}}$PCL formulas.

\begin{proposition}
	Let $\zeta, \zeta_1, \zeta_2, \zeta_3 \in PCL(D,P)$ and $d\in D$. Then  
	\begin{enumerate}[(i)]
		\item $\zeta\uplus 0 \equiv 0 \equiv 0\uplus \zeta. $  
		
	\end{enumerate}
	
\noindent If $\otimes$ is commutative, then  

\begin{itemize}
	\item[(ii)] $\zeta_1 \uplus \zeta_2 \equiv \zeta_2 \uplus \zeta_1.$ 
\end{itemize}

\noindent If $D$ is associative and $\oplus$-distributive, then

\begin{itemize}
	\item[(iii)] $ (\zeta_1 \uplus \zeta_2) \uplus \zeta_3 \equiv \zeta_1 \uplus (\zeta_2 \uplus \zeta_3) .$
\end{itemize}

\noindent If $D$ is left-$\oplus$-distributive, then 
\begin{itemize}
	\item[(iv)] $\zeta\otimes(\zeta_1\oplus\zeta_2) \equiv (\zeta\otimes\zeta_1)\oplus(\zeta\otimes\zeta_2).$
\end{itemize} 

\noindent If $D$ is right-$\oplus$-distributive, then 
\begin{itemize}
	\item[(v)] $(\zeta_1\oplus\zeta_2) \otimes \zeta \equiv (\zeta_1\otimes\zeta)\oplus(\zeta_2\otimes\zeta).$
\end{itemize}

\end{proposition}

\begin{proof}
	For every $\gamma \in C(P)$ we have 
\begin{enumerate}[(i)]
	\item \begin{align*}
	\left\Vert \zeta \uplus 0 \right\Vert (\gamma) & = \bigoplus_{\gamma_1 \cupdot \gamma_2 = \gamma }  \left( \left\Vert \zeta \right\Vert (\gamma_1) \otimes \left\Vert 0 \right\Vert (\gamma_2) \right) \\ & = 0 \\ & = \left\Vert 0 \uplus \zeta  \right\Vert (\gamma).
	\end{align*}

	\item 
	\begin{align*}
	\left\Vert \zeta_1 \uplus \zeta_2 \right\Vert (\gamma) & = \bigoplus_{\gamma_1 \cupdot \gamma_2 = \gamma} (\left\Vert \zeta_1\right\Vert (\gamma_1) \otimes \left\Vert \zeta_2 \right\Vert(\gamma_2)) \\ & = \bigoplus_{ \gamma_1 \cupdot \gamma_2 = \gamma} (\left\Vert \zeta_2 \right\Vert(\gamma_2) \otimes \left\Vert \zeta_1\right\Vert (\gamma_1) ) \\ & = \left\Vert \zeta_2 \uplus \zeta_1 \right\Vert(\gamma) ,
	\end{align*}
	where the second equality holds since $\otimes$ is commutative.
	\item \begin{align*}
	\left\Vert (\zeta_1 \uplus \zeta_2) \uplus \zeta_3 \right\Vert (\gamma) & = \bigoplus_{ \gamma^\prime \cupdot \gamma_3 = \gamma } \left(  \left\Vert \zeta_1 \uplus \zeta_2 \right\Vert(\gamma^\prime) \otimes \left\Vert \zeta_3 \right\Vert(\gamma_3)  \right)  \\ & = \bigoplus_{ \gamma^\prime \cupdot \gamma_3 = \gamma } \left( \left(  \bigoplus_{\gamma_1 \cupdot \gamma_2 = \gamma^\prime  } (\left\Vert \zeta_1\right\Vert(\gamma_1) \otimes \left\Vert \zeta_2 \right\Vert (\gamma_2) )  \right) \otimes \left\Vert \zeta_3 \right\Vert(\gamma_3)  \right)  \\ & = \bigoplus_{ \gamma_1 \cupdot \gamma_2  \cupdot \gamma_3 = \gamma } \left(  (\left\Vert \zeta_1\right\Vert(\gamma_1) \otimes \left\Vert \zeta_2 \right\Vert (\gamma_2) )   \otimes \left\Vert \zeta_3 \right\Vert(\gamma_3)  \right)  \\ & = \bigoplus_{ \gamma_1 \cupdot \gamma_2  \cupdot \gamma_3 = \gamma } \left(  \left\Vert \zeta_1\right\Vert(\gamma_1) \otimes \left( \left\Vert \zeta_2 \right\Vert (\gamma_2)    \otimes \left\Vert \zeta_3  \right\Vert(\gamma_3) \right) \right) \\ & = \bigoplus_{ \gamma_1 \cupdot \gamma^\prime =\gamma  } \left\Vert \zeta_1 \right\Vert(\gamma_1) \otimes \left(  \bigoplus_{ \gamma_2 \cupdot \gamma_3 = \gamma^\prime }  \left(  \left\Vert \zeta_2 \right\Vert (\gamma_2) \otimes \left\Vert \zeta_3 \right\Vert(\gamma_3) \right) \right)  \\ & = \bigoplus_{ \gamma_1 \cupdot \gamma^\prime =\gamma  } \left\Vert \zeta_1 \right\Vert(\gamma_1) \otimes  \left\Vert \zeta_2 \uplus \zeta_3 \right\Vert (\gamma^\prime)  \\ & = \left\Vert \zeta_1 \uplus (\zeta_2 \uplus \zeta_3) \right\Vert(\gamma)
	\end{align*}

	\noindent where the third and fifth equality hold since $D$ is $\oplus$-distributive and the fourth one since $D$ is associative.

	\item \begin{align*}
	\left\Vert \zeta \otimes (\zeta_1 \oplus \zeta_2) \right\Vert (\gamma) & = \left\Vert \zeta \right\Vert (\gamma) \otimes \left\Vert \zeta_1 \oplus \zeta_2 \right\Vert (\gamma) \\ & =  \left\Vert \zeta \right\Vert (\gamma) \otimes \left( \left\Vert \zeta_1 \right\Vert(\gamma) \oplus \left\Vert  \zeta_2 \right\Vert (\gamma)  \right) \\ & = \left( \left\Vert \zeta \right\Vert (\gamma) \otimes \left\Vert \zeta_1 \right\Vert(\gamma) \right) \oplus \left( \left\Vert \zeta \right\Vert (\gamma) \otimes \left\Vert \zeta_2 \right\Vert(\gamma) \right) \\ & = \left\Vert\zeta \otimes \zeta_1 \right\Vert(\gamma) \oplus \left\Vert\zeta \otimes \zeta_2 \right\Vert(\gamma) \\ & = \left\Vert \left( \zeta \otimes \zeta_1 \right) \oplus \left( \zeta \otimes \zeta_2 \right) \right\Vert(\gamma), \end{align*}
	where the third equality holds since $D$ is left-$\oplus$-distributive. 
\item The proof is similar to the one of $(iv)$.\end{enumerate} \end{proof}

\medskip
\begin{proposition}
	Let $\zeta \in PCL(D,P)$ with $\zeta=d \in D$. If $D$ left-${\val}$-distributive, then \[ *\zeta  \equiv d.\]  
\end{proposition}
\begin{proof}
	For every $\gamma \in C(P)$ we have 
	\begin{align*}
	\left\Vert *\zeta \right\Vert (\gamma)  & = \bigoplus_{n>0} \  \bigoplus_{\gamma_1 \cupdot ... \cupdot \gamma_n=\gamma} {\val}(\left\Vert \zeta \right\Vert (\gamma_1), ... , \left\Vert \zeta \right\Vert (\gamma_n)).
	\end{align*}
	
	\noindent Let $\gamma=\{ a_1, \dots, a_s \}$ where $s\in \mathbb{N}$. Then, we get the following 
	\begin{align*}
	\left\Vert *\zeta \right\Vert (\gamma)  & = \underset{n\in \{ 1, \dots, s\}}{\bigoplus} \  \underset{\gamma_1 \cupdot ... \cupdot \gamma_n=\gamma}{\bigoplus} {\val}(\left\Vert \zeta \right\Vert (\gamma_1), ... , \left\Vert \zeta \right\Vert (\gamma_n))\\  & ={\val}(d)   \oplus {\val}(d,d) \oplus ...\oplus {\val}(\overbrace{d,...,d}^{ s \text{ times} }) \\ & =( d\otimes {\val}(1)) \oplus (d\otimes {\val}(1,1)) \oplus ... \oplus (d\otimes {\val}(1,...,1)) \\ & = (d\otimes 1) \oplus (d\otimes 1) \oplus ... \oplus (d\otimes 1) \\ & = d\oplus ...\oplus d \\ & = d
	\end{align*}
	where the second and the last equalities hold since $D$ is idempotent and the third one since $D$ is left-${\val}$-distributive. 
\end{proof}

\begin{definition}\label{val_distri_addi}
Let $(D, \oplus, {\val}, 0)$ be a valuation monoid. The valuation function ${\val}$ is called left-$\oplus$-preservative whenever the following holds:
	\[ {\val}(d_1 \oplus d_2 , d ) = {\val}(d_1,d) \oplus   {\val}(d_2,d) \]
	for any $d, d_1, d_2\in D.$ Analogously, ${\val}$ is called right-$\oplus$-preservative if
	\[  {\val}(d, d_1 \oplus d_2) = {\val}(d, d_1) \oplus {\val}(d, d_2) \]
	for any $d, d_1, d_2\in D.$ If ${\val}$ is both left- and right-$\oplus$-preservative, then it is called $\oplus$-preservative.
\end{definition}

By a straightforward calculation we can show the next proposition. 

\begin{proposition}
	Let $D$ be a valuation monoid. If ${\val}$ is $\oplus$-preservative, then  
	\[ {\val} \left( \bigoplus_{ i\in I } d_i, \bigoplus_{j\in J}d_j^\prime  \right) = \bigoplus_{i\in I, j\in J} {\val}\left(d_i, d_j^\prime\right)  \]
	where $I, J$ are finite index sets and $d_i, d_j^\prime \in D$ for every $i\in I$ and $j\in J$.
\end{proposition}

Considering Definition \ref{val_distri_addi} and the pv-monoids $(\mathbb{R}\cup\{-\infty\}, \max, {\avg}, +, -\infty, 0)$ and $(\mathbb{R}  $  $ \cup $ $\{+\infty\}, $ $\min,$ $ {\avg}, +, +\infty, 0)$, ${\avg}$ is $\oplus$-preservative in both cases. \hide{Consider the pv-monoid $\left( \mathbb{R}\cup \{ -\infty \}, \max, disc_{\lambda}, \right. $ $ \left. + -\infty, 0  \right)$, with $\lambda \in \mathbb{R}, $ $\lambda<1$ and $disc_{\lambda}(d_0, \dots, d_n) = \sum_{i=0}^{n}\lambda^i d_i $. Then $disc_{\lambda}$ is $\oplus$-preservative.}

\hide{ Let for example the pv-monoid $(\mathbb{R}\cup\{-\infty\}, \max, {\avg}, +, -\infty, 0)$ and the values $d, d_1, d_2 \in \mathbb{R}\cup \{ -\infty \}$, then:
\begin{align*}
{\avg}\left( \max\{ d_1, d_2 \} , d\right) & = \frac{\max\{ d_1, d_2 \} + d}{2} = \max\left\{ \frac{d_1+d}{2}, \frac{d_2+d}{2} \right\} \\ & = \max \left\{ {\avg}(d_1, d), {\avg}(d_2, d) \right\}.
\end{align*}
Analogously, we prove that ${\avg}$ is right-$\oplus$-preservative and conclude that ${\avg}$ is $\oplus$-preservative. }

\begin{proposition}
	Let $\zeta \in PCL(D,P)$. If $D$ is $\oplus$-preservative, then \[\sim (*\zeta) \equiv * (\sim \zeta).\]
\end{proposition}

\begin{proof}
	
	Let $\gamma \in C(P)$. Then 
	\begin{align*}
	\left\Vert *(\sim \zeta) \right\Vert (\gamma)  & = \underset{n>0}{\bigoplus} \  \underset{\gamma_1 \cupdot ... \cupdot \gamma_n=\gamma}{\bigoplus} {\val}(\left\Vert \sim \zeta \right\Vert (\gamma_1), ... , \left\Vert \sim \zeta \right\Vert (\gamma_n)) \\ & =  \underset{n>0}{\bigoplus} \  \underset{\gamma_1 \cupdot ... \cupdot \gamma_n=\gamma}{\bigoplus} {\val}\left( \underset{\gamma_1^\prime\subseteq \gamma_1}{\bigoplus} \left\Vert  \zeta \right\Vert (\gamma_1^\prime), ... \ , \underset{\gamma_n^\prime\subseteq \gamma_n}{\bigoplus} \left\Vert  \zeta \right\Vert (\gamma_n^\prime)\right) \\ & = \underset{n>0}{\bigoplus} \  \underset{\gamma_1 \cupdot ... \cupdot \gamma_n=\gamma}{\bigoplus}\  \underset{\gamma_1^\prime\subseteq \gamma_1}{\bigoplus} ... \ \underset{\gamma_n^\prime\subseteq \gamma_n}{\bigoplus}{\val}(  \left\Vert  \zeta \right\Vert (\gamma_1^\prime), ... \ , \left\Vert  \zeta \right\Vert (\gamma_n^\prime)) \\ & = \underset{n>0}{\bigoplus} \  \underset{\gamma_1 \cupdot ... \cupdot \gamma_n \subseteq \gamma}{\bigoplus}\ {\val}(  \left\Vert  \zeta \right\Vert (\gamma_1), ... \ , \left\Vert  \zeta \right\Vert (\gamma_n)) \\ & = \underset{n>0}{\bigoplus} \  \underset{\gamma^\prime \subseteq \gamma}{\bigoplus} \  \underset{\gamma_1 \cupdot ... \cupdot \gamma_n = \gamma^\prime}{\bigoplus}\ {\val}(  \left\Vert  \zeta \right\Vert (\gamma_1), ... \ , \left\Vert  \zeta \right\Vert (\gamma_n)) \\ & = \underset{\gamma^\prime \subseteq \gamma}{\bigoplus} \ \underset{n>0}{\bigoplus} \   \  \underset{\gamma_1 \cupdot ... \cupdot \gamma_n = \gamma^\prime}{\bigoplus}\ {\val}(  \left\Vert  \zeta \right\Vert (\gamma_1), ... \ , \left\Vert  \zeta \right\Vert (\gamma_n)) \\ & = \underset{\gamma^\prime \subseteq \gamma}{\bigoplus} \ \left\Vert *\zeta \right\Vert (\gamma^\prime) \\ & = \left\Vert \sim (*\zeta )\right\Vert (\gamma) 
	\end{align*}
	where the third equality holds since $D$ is $\oplus$-preservative and the next equalities due to the commutativity of $\oplus.$ 
\end{proof}

\begin{proposition}\label{uplus_over_oplus_right}
	Let $\zeta, \zeta_1, \zeta_2 \in PCL(D,P)$. If $D$ is left-$\oplus$-distributive, then

		\[\zeta \uplus(\zeta_1 \oplus\zeta_2) \equiv (\zeta \uplus \zeta_1) \oplus (\zeta \uplus \zeta_2). \]

\end{proposition}

\begin{proof}
	For every $\gamma \in C(P)$ we have
	\begin{align*}
	\left \Vert \zeta \uplus(\zeta_1 \oplus\zeta_2) \right \Vert (\gamma)  & = \bigoplus_{\gamma^\prime \cupdot \gamma^{\prime \prime} = \gamma } \left(  \left\Vert \zeta \right\Vert(\gamma^\prime) \otimes \left\Vert \zeta_1 \oplus \zeta_2 \right\Vert (\gamma^{\prime \prime}) \right) \\ & = \bigoplus_{\gamma^\prime \cupdot \gamma^{\prime \prime} = \gamma } \left(  \left\Vert \zeta \right\Vert(\gamma^\prime) \otimes  \left(  \left\Vert \zeta_1 \right\Vert(\gamma^{\prime \prime})  \oplus \left\Vert \zeta_2 \right\Vert(\gamma^{\prime \prime})   \right)   \right) \\ & = \bigoplus_{\gamma^\prime \cupdot \gamma^{\prime \prime} = \gamma }\left(  \left\Vert \zeta \right\Vert (\gamma^\prime ) \otimes \left\Vert \zeta_1 \right\Vert(\gamma^{\prime \prime } )    \right) \oplus \left(  \left\Vert \zeta \right\Vert (\gamma^\prime ) \otimes \left\Vert \zeta_2 \right\Vert(\gamma^{\prime \prime } ) \right) \\ & = \bigoplus_{\gamma^\prime \cupdot \gamma^{\prime \prime} = \gamma }\left(  \left\Vert \zeta \right\Vert (\gamma^\prime ) \otimes \left\Vert \zeta_1 \right\Vert(\gamma^{\prime \prime } )    \right) \oplus \bigoplus_{\gamma^\prime \cupdot \gamma^{\prime \prime} = \gamma }\left(  \left\Vert \zeta \right\Vert (\gamma^\prime ) \otimes \left\Vert \zeta_2 \right\Vert(\gamma^{\prime \prime } )    \right) \\ & = \left\Vert \zeta \uplus \zeta_1 \right\Vert(\gamma) \oplus \left\Vert \zeta \uplus \zeta_2 \right\Vert(\gamma) \\ & = \left\Vert (\zeta \uplus \zeta_1) \oplus ( \zeta \uplus \zeta_2) \right\Vert(\gamma)
	\end{align*}
	where the third equality holds since $D$ is left-$\oplus$-distributive and the fourth one since $\oplus$ is commutative. 
	\end{proof}

\medskip

\begin{proposition}\label{uplus_over_oplus_left}
	Let $\zeta, \zeta_1, \zeta_2 \in PCL(D,P)$. If $D$ is right-$\oplus$-distributive, then
	
	\[(\zeta_1 \oplus\zeta_2)  \uplus \zeta \equiv (\zeta_1 \uplus \zeta) \oplus (\zeta_2 \uplus \zeta). \]
	
\end{proposition}

\begin{proof}
	The proof is similar to the one of Proposition \ref{uplus_over_oplus_right}.
\end{proof}

\medskip

Next, we show a special case when $\otimes $ distributes over $\uplus$. In general $\otimes$ does not distribute over $\uplus$. For example, let $P=\{ p,q \}$ and the w$_{\text{pvm}}$PCL formulas $\zeta, \zeta_1, \zeta_2$, where $\zeta=2$ and $\zeta_1=\zeta_2 = 1$. If we consider the set $\gamma= \{ \{p\}, \{q\} \}$ and the pv-monoid $(\mathbb{R}\cup\{-\infty\}, \max, {\avg}, +, -\infty, 0)$, then it is easy to show that $\left\Vert \zeta \otimes (\zeta_1 \uplus \zeta_2) \right\Vert(\gamma)  \not =  \left\Vert (\zeta \otimes \zeta_1) \uplus (\zeta \otimes \zeta_2) \right\Vert(\gamma)$. Hence, $\zeta \otimes(\zeta_1 \uplus \zeta_2) \not \equiv (\zeta \otimes \zeta_1) \uplus (\zeta \otimes \zeta_2).$ However, this is not the case when $\zeta$ is a PIL formula and $D$ is left-$\oplus$-distributive.
\begin{proposition}
	\label{conj_over_coal} Let $\phi$ be a \emph{PIL} formula over $P$ and $\zeta_{1},\zeta
	_{2}\ \in PCL(D,P)$. If $D$ is left-$\oplus$-distributive, then
	\[
	\phi\otimes(\zeta_{1}\uplus\zeta_{2})\equiv(\phi\otimes\zeta_{1})\uplus
	(\phi\otimes\zeta_{2}).
	\]
	
\end{proposition}

\begin{proof}
	For every $\gamma\in C(P)$ we have%
	\begin{align*}
	\left\Vert \phi\otimes(\zeta_{1}\uplus\zeta_{2})\right\Vert (\gamma)  &
	=\left\Vert \phi\right\Vert (\gamma)\otimes\left\Vert \zeta_{1}\uplus\zeta
	_{2}\right\Vert (\gamma)\\
	&  =\left\Vert \phi\right\Vert (\gamma)\otimes\left(  \underset{\gamma
		=\gamma_{1}\cupdot \gamma_{2}}{\bigoplus}\left(  \left\Vert \zeta_{1}\right\Vert
	(\gamma_{1})\otimes\left\Vert \zeta_{2}\right\Vert (\gamma_{2})\right)
	\right) \\
	&  =\underset{\gamma=\gamma_{1}\cupdot \gamma_{2}}{\bigoplus}\left(  \left\Vert
	\phi\right\Vert (\gamma)\otimes(\left\Vert \zeta_{1}\right\Vert (\gamma
	_{1})\otimes\left\Vert \zeta_{2}\right\Vert (\gamma_{2}))\right)  .
	\end{align*}
	We distinguish two cases.
	
	\begin{itemize}
		\item $\Vert\phi\Vert(\gamma)=1$. Then by definition, $\gamma\models\phi$ or $\alpha\models_i \phi$ for every $\alpha \in \gamma$. Hence, $\gamma^{\prime}\models\phi$ for every $\gamma^\prime \subseteq \gamma$, and subsequently $\Vert\phi
		\Vert(\gamma^{\prime})=1$\ for every $\gamma^{\prime}\subseteq\gamma$.
		Therefore, we get
		\begin{align*}
		&  \underset{\gamma=\gamma_{1}\cupdot \gamma_{2}}{\bigoplus}\left(  \left\Vert
		\phi\right\Vert (\gamma)\otimes(\left\Vert \zeta_{1}\right\Vert (\gamma
		_{1})\otimes\left\Vert \zeta_{2}\right\Vert (\gamma_{2}))\right)  \text{ }\\
		&  =\underset{\gamma=\gamma_{1}\cupdot \gamma_{2}}{\bigoplus}\left(  \left\Vert
		\phi\right\Vert (\gamma_{1})\otimes\left\Vert \zeta_{1}\right\Vert (\gamma
		_{1})\otimes\Vert\phi\Vert(\gamma_{2})\otimes\left\Vert \zeta_{2}\right\Vert
		(\gamma_{2})\right) \\
		&  =\underset{\gamma=\gamma_{1}\cupdot \gamma_{2}}{\bigoplus}\left(  \Vert
		\phi\otimes\zeta_{1}\Vert(\gamma_{1})\otimes\Vert\phi\otimes\zeta_{2}%
		\Vert(\gamma_{2})\right) \\
		&  =\Vert(\phi\otimes\zeta_{1})\uplus(\phi\otimes\zeta_{2})\Vert(\gamma).
		\end{align*}

		\item $\Vert\phi\Vert(\gamma)=0$. Hence $\gamma\not \models \phi$, i.e., there
		is an $a\in\gamma$ such that $a\not \models _{i}\phi$. This in turn implies
		that $\gamma^{\prime}\not \models \phi$ for every $\gamma^{\prime}%
		\subseteq\gamma$ with $a\in \gamma^{\prime}$. Therefore, we get
		\[
		\underset{\gamma=\gamma_{1}\cupdot \gamma_{2}}{\bigoplus}\left(  \left\Vert
		\phi\right\Vert (\gamma)\otimes(\left\Vert \zeta_{1}\right\Vert (\gamma
		_{1})\otimes\left\Vert \zeta_{2}\right\Vert (\gamma_{2}))\right)  =0
		\]
		and%
		\[
		\underset{\gamma=\gamma_{1}\cupdot \gamma_{2}}{\bigoplus}\left(  \left\Vert
		\phi\right\Vert (\gamma_{1})\otimes\left\Vert \zeta_{1}\right\Vert (\gamma
		_{1})\otimes\Vert\phi\Vert(\gamma_{2})\otimes\left\Vert \zeta_{2}\right\Vert
		(\gamma_{2})\right)  =0,
		\]
		i.e.,
		\[
		\left\Vert \phi\otimes(\zeta_{1}\uplus\zeta_{2})\right\Vert (\gamma
		)=0=\Vert(\phi\otimes\zeta_{1})\uplus(\phi\otimes\zeta_{2})\Vert(\gamma)
		\]
		and this concludes our proof. 
	\end{itemize}
\end{proof}

\section{Full normal form for w$_{\text{pvm}}$PCL formulas }

In this section, we show that for every w$_{\text{pvm}}$PCL formula $\zeta \in PCL(D,P)$, where $D$ is a pv-monoid satisfying specific properties, we can effectively construct an equivalent formula of a special form which is called \textit{full normal form}. For this, we will use corresponding results from \cite{Ma:Co} and \cite{Pa:On}. More precisely, for every PCL formula $f$ over $P$ we can effectively construct a unique equivalent PCL formula of the form $true$\footnote{Following \cite{Pa:Wei} we consider $true$ as a full normal form.} or $\bigsqcup_{i\in I} \sum_{j\in J_{i}} m_{i,j}$ (cf. Theorem 4.43 in \cite{Ma:Co}), and for every weighted PCL formula $\zeta$ over $P$ and a commutative semiring $(K, \oplus, \otimes, 0,1)$ we can construct a unique equivalent weighted PCL formula of the form $k$ or $\bigoplus_{i\in I} \left( k_i \otimes \sum_{j\in J_i} m_{i,j}  \right)$ (cf. Theorem 1 in \cite{Pa:On} and Theorem 25 in \cite{Pa:Wei}). The index sets $I$ and $J_i$, for every $i\in I$, are finite, $k$ and $k_i \in K$ and $m_{i,j}$'s are full monomials over $P.$  We show that we can also effectively build a unique full normal form for every w$_{\text{pvm}}$PCL formula over $P$ and a pv-monoid $D$ satisfying specific properties shown below. Uniqueness is up to the equivalence relation. Lastly, we show that the equivalence problem of w$_{\text{pvm}}$PCL formulas is decidable. 

\begin{definition}
	A {\normalfont w$_{\text{pvm}}$PCL} formula $\zeta \in$ $PCL(D,P)$ is said to be in full normal form if either
	\begin{enumerate}[1.]
		\item $\zeta = d$, with $d\in D$, or 
		\item 	there are finite index sets I and $J_i$ for every $i \in I,d_i \in D,$ and full monomials $m_{i,j}$ for every $i \in I$ and $ j \in J_i $ such that $\zeta=\bigoplus_{i\in I}\left(d_i \otimes \sum_{j\in J_i} m_{i,j}\right).$
	\end{enumerate}
\end{definition}

Following \cite{Pa:Wei}, for every full normal form we can construct an equivalent one satisfying the subsequent statements:
\begin{enumerate}[(i)]
	\item $j \not = j^\prime$ implies $m_{i,j} \not \equiv m_{i,j^\prime}$ for every $i\in I$, $j,j^\prime \in J_i$, and 
	\item $i\not = i^\prime $ implies $\sum_{j\in J_i} m_{i,j} \not \equiv \sum_{j\in J_{i^\prime}} m_{i^\prime, j} $ for every $i, i^\prime \in I$.
\end{enumerate}

\noindent By Lemma 1 in \cite{Pa:On}, if $m_{i,j} \equiv m_{i, j^\prime}$ for some $j\not = j^\prime$, then we get $m_{i,j} + m_{i,j^\prime} \equiv m_{i,j}$. So, we replace $m_{i,j} + m_{i,j^\prime}$ by $m_{i,j}$. For the second case, let $\sum_{j\in J_i} m_{i,j} \equiv \sum_{j\in J_{i^\prime}} m_{i^\prime, j}$ for some $i\not = i^\prime$. Then, we replace $ \left( d_i \otimes \sum_{j\in J_i} m_{i,j} \right)$ $ \oplus\left( d_{i^\prime} \otimes \sum_{j\in J_{i^\prime}} m_{i^\prime,j} \right)  $ by its equivalent formula $ \left( d_i \oplus d_{i^\prime} \right) \otimes \sum_{j\in J_i} m_{i,j}  $. In the sequel, we assume that every full normal form satisfies Statements (i) and (ii).

For the construction of the full normal form of every $\zeta \in PCL(D,P)$ we shall need the next results. The proofs of Lemmas \ref{fullno5} and \ref{fullno8}, Propositions \ref{fullnoPCL} and \ref{w_coal_unw}, and Theorem \ref{equa_decid} are similar to the corresponding ones in \cite{Pa:Wei}.

\begin{lemma} 
	\label{fullno5}Let $J$ be an index set and $m_{j}$ full monomials for every
	$j\in J$. Then, there exists a unique $\overline{\gamma}\in C(P)$ such that
	for every $\gamma\in C(P)$ we have $\left\Vert \sum\nolimits_{j\in J}%
	m_{j}\right\Vert (\gamma)=1$ if $\gamma=\overline{\gamma}$ and $\left\Vert
	\sum\nolimits_{j\in J}m_{j}\right\Vert (\gamma)=0$, otherwise.
\end{lemma}

\begin{proof}
	For every full monomial $m_{j}$, $j\in J$, there exists a unique interaction $a_{j}$ such
	that $a_{j}\models_{i}m_{j}$. Then, it is straightforward to show that
	$\overline{\gamma}=\{a_{j}\mid j\in J\}$ satisfies our claim.
\end{proof}

\begin{proposition}
	\label{fullnoPCL} Let $f$ be a $PCL$ formula over $P$ and $D$ a pv-monoid. Then there exist finite
	index sets $I$ and $J_{i}$ for every $i\in I$, and full monomials $m_{i,j}$
	for every $i\in I$ and $j\in J_{i}$ such that
	\[
	f\equiv\bigoplus_{i\in I}\sum_{j\in J_{i}}m_{i,j}\equiv\bigoplus_{i\in
		I}\left(  1\otimes\sum_{j\in J_{i}}m_{i,j}\right)  .
	\]

\end{proposition}

\begin{proof}
	By Theorem 4.43. in \cite{Ma:Co} there exists a unique full normal form such that
	$f\equiv\bigsqcup_{i\in I}\sum_{j\in J_{i}}m_{i,j}$, where $m_{i,j}$ are full
	monomials over $P$. By Lemma \ref{fullno5}, for every $i\in I$ there
	exists a unique $\overline{\gamma}_{i}\in C(P)$, such that for every $\gamma\in
	C(P)$ we have $\left\Vert \sum_{j\in J_{i}}m_{i,j}\right\Vert (\gamma)=1$ if
	$\gamma=\overline{\gamma}_{i}$ and $\left\Vert \sum_{j\in J_{i}}%
	m_{i,j}\right\Vert (\gamma)=0$ otherwise. Then,\ %
	\begin{align*}
	\left\Vert f\right\Vert (\gamma)  &  =\left\{
	\begin{array}
	[c]{ll}%
	1 & \textnormal{ if }\gamma\models\bigsqcup_{i\in I}\sum_{j\in J_{i}}m_{i,j}\\
	0 & \textnormal{ otherwise}%
	\end{array}
	\right. \\
	&  =\left\{
	\begin{array}
	[c]{ll}%
	1 & \textnormal{ if }\gamma=\overline{\gamma}_{i}\text{ for some }i\in I\\
	0 & \textnormal{ otherwise}%
	\end{array}
	\right. \\
	&  =\left\{
	\begin{array}
	[c]{ll}%
	1 & \textnormal{ if }\left\Vert \sum_{j\in J_{i}}m_{i,j}\right\Vert (\gamma)=1\text{
		for some }i\in I\\
	0 & \textnormal{ otherwise}%
	\end{array}
	\right. \\
	&  =\left\{
	\begin{array}
	[c]{ll}%
	1 & \textnormal{ if }\left\Vert \bigoplus_{i\in I}\sum_{j\in J_{i}}m_{i,j}\right\Vert
	(\gamma)=1\text{ }\\
	0 & \textnormal{ otherwise}%
	\end{array}
	\right. \\
	&  =\left\{
	\begin{array}
	[c]{ll}%
	1 & \textnormal{ if }\left\Vert \bigoplus_{i\in I}\left(  1\otimes\sum_{j\in J_{i}%
	}m_{i,j}\right)  \right\Vert (\gamma)=1\text{ }\\
	0 & \textnormal{ otherwise}%
	\end{array}
	\right.
	\end{align*}
	Hence, we proved that $f\equiv\bigoplus_{i\in I}\sum_{j\in J_{i}%
	}m_{i,j}\equiv\bigoplus_{i\in I}\left(  1\otimes\sum_{j\in J_{i}}%
	m_{i,j}\right)  $, as required.
\end{proof}

\begin{lemma}
	\label{fullno8} Let $m_{i},m_{j}^{\prime}$ be full monomials for every $i\in
	I$ and $j\in J$. Then,
	\[
	\left(  \underset{i\in I}{\sum}m_{i}\right)  \otimes\left(  \underset{j\in
		J}{\sum}m_{j}^{\prime}\right)  \equiv\left\{
	\begin{array}
	[c]{ll}%
	\underset{i\in I}{\sum}m_{i} & \textnormal{ if }\underset{i\in I}{\sum}m_{i}%
	\equiv\underset{j\in J}{\sum}m_{j}^{\prime},\\
	0 & \textnormal{ otherwise.}%
	\end{array}
	\right.
	\]
	
\end{lemma}

\begin{proof}
	By Lemma \ref{fullno5} there exist $\overline{\gamma},\overline{\gamma
	}^{\prime}\in C(P)$ such that for every $\gamma\in C(P)$ the following holds: $\left\Vert
	\sum\nolimits_{i\in I}m_{i}\right\Vert(\gamma) =1$ if $\gamma=\overline
	{\gamma}$ and $\left\Vert \sum\nolimits_{i\in I}m_{i}\right\Vert (\gamma)=0$
	otherwise, and $\left\Vert \sum\nolimits_{j\in J}m_{j}^{\prime}\right\Vert
	(\gamma)=1$ if $\gamma=\overline{\gamma}^{\prime}$ and $\left\Vert
	\sum\nolimits_{j\in J}m_{j}^{\prime}\right\Vert (\gamma)=0$ otherwise.
	Therefore, for every $\gamma\in C(P)$ we get
	\begin{align*}
	\left\Vert \left(  \underset{i\in I}{\sum}m_{i}\right)  \otimes\left(
	\underset{j\in J}{\sum}m_{j}^{\prime}\right)  \right\Vert (\gamma)  &
	=\left\Vert \underset{i\in I}{\sum}m_{i}\right\Vert (\gamma)\otimes\left\Vert
	\underset{j\in J}{\sum}m_{j}^{\prime}\right\Vert (\gamma)\\
	&  =\left\{
	\begin{array}
	[c]{ll}%
	1\otimes1 & \textnormal{ if }\underset{i\in I}{\sum}m_{i}\equiv\underset{j\in J}%
	{\sum}m_{j}^{\prime}\text{ and }\gamma=\overline{\gamma}=\overline{\gamma
	}^{\prime}\\
	0 & \textnormal{ otherwise}%
	\end{array}
	\right. \\
	&  =\left\{
	\begin{array}
	[c]{ll}%
	1 & \textnormal{ if }\underset{i\in I}{\sum}m_{i}\equiv\underset{j\in J}{\sum}%
	m_{j}^{\prime}\text{ and }\gamma=\overline{\gamma}=\overline{\gamma}^{\prime
	}\\
	0 & \textnormal{ otherwise}%
	\end{array}
	\right. \\
	&  =\left\{
	\begin{array}
	[c]{ll}%
	\left\Vert \underset{i\in I}{\sum}m_{i}\right\Vert (\gamma) & \textnormal{ if
	}\underset{i\in I}{\sum}m_{i}\equiv\underset{j\in J}{\sum}m_{j}^{\prime}\\
	0 & \textnormal{ otherwise}%
	\end{array}
	\right.
	\end{align*}
	which concludes our claim.
\end{proof}

\begin{proposition}\label{otim_coale}
	Let $d_1, d_2 \in D$ and $\zeta_1, \zeta_2\in PCL(D,P)$. If $D$ is left-$\oplus$-distributive and $\otimes $ is commutative and associative, then
	\[ \left(  d_1 \otimes \zeta_1   \right) \uplus \left(  d_2 \otimes \zeta_2 \right) \equiv d_1 \otimes d_2 \otimes (\zeta_1 \uplus \zeta_2). \]
\end{proposition}

\begin{proof}
	For every $\gamma \in C(P)$ we have 
	\begin{align*}
	\left\Vert \left(  d_1 \otimes \zeta_1   \right) \uplus \left(  d_2 \otimes \zeta_2 \right)  \right\Vert(\gamma) & = \bigoplus_{\gamma_1 \cupdot \gamma_2 =\gamma} \left(  \left\Vert d_1\otimes \zeta_1 \right\Vert(\gamma_1) \otimes \left\Vert d_2 \otimes \zeta_2 \right\Vert(\gamma_2)  \right) \\ & = \bigoplus_{\gamma_1 \cupdot \gamma_2 =\gamma} \left( \left( d_1 \otimes \left\Vert \zeta_1 \right\Vert(\gamma_1) \right) \otimes \left( d_2 \otimes \left\Vert \zeta_2 \right\Vert(\gamma_2)   \right) \right)\\ & = \bigoplus_{\gamma_1 \cupdot \gamma_2 =\gamma}  (d_1 \otimes d_2 \otimes  \left\Vert \zeta_1 \right\Vert(\gamma_1) \otimes \left\Vert \zeta_2 \right\Vert(\gamma_2) ) \\ & = (d_1 \otimes d_2) \otimes  \bigoplus_{\gamma_1 \cupdot \gamma_2 =\gamma}  \left( \left\Vert \zeta_1 \right\Vert(\gamma_1) \otimes \left\Vert \zeta_2 \right\Vert(\gamma_2) \right) \\ & = (d_1 \otimes d_2) \otimes \left\Vert\zeta_1 \uplus \zeta_2 \right\Vert(\gamma) \\ & =  \left\Vert  (d_1 \otimes d_2) \otimes \zeta_1 \uplus \zeta_2 \right\Vert(\gamma),
	\end{align*}
	where the third equality holds by the commutativity and associativity of $\otimes$ and the fourth one since $D$ left-$\oplus$-distributive.
\end{proof}

\begin{proposition}\label{w_coal_unw}
	Let $m_i, m_j^\prime $ be full monomials for every $i \in I$ and $j\in J$. Then 
	\[  \left(  \sum_{i\in I} m_i \right) \uplus \left(  \sum_{j\in J} m_j^\prime \right)  \equiv \left\{  \begin{array}{l l}
	 \sum_{i\in I} m_i +   \sum_{j\in J} m_j^\prime & \text{if } m_i \not \equiv m_j^\prime \text{ for every } i\in I \text{ and } j\in J \\
	 0 & \text{otherwise.}
	\end{array} \right.   \]
	
\end{proposition}

\begin{proof}
	By Lemma \ref{fullno5} there exist $\overline{\gamma}, \overline{\gamma}^\prime \in C(P) $ such that for every $\gamma\in C(P)$ we have $\left\Vert \sum_{i\in I}m_i \right\Vert(\gamma) = 1$ if $\gamma = \overline{\gamma}$ and $\left\Vert \sum_{i\in I}m_i \right\Vert(\gamma) = 0$ otherwise, and $\left\Vert \sum_{j\in J}m_j^\prime \right\Vert(\gamma) = 1$ if $\gamma = \overline{\gamma}^\prime$ and $\left\Vert \sum_{j\in J}m_j^\prime \right\Vert(\gamma) = 0$ otherwise. If $\overline{\gamma} \cap \overline{\gamma}^\prime = \emptyset$, for every $\gamma\in C(P)$ we get
	\begin{align*}
	\left\Vert \left(  \sum_{i\in I} m_i \right) \uplus \left(  \sum_{j\in J} m_j^\prime \right) \right\Vert (\gamma) & = \bigoplus_{\gamma_1 \cupdot \gamma_2 = \gamma} \left(  \left\Vert \sum_{i\in I } m_i \right\Vert(\gamma_1) \otimes \left\Vert \sum_{j\in J } m_j^\prime \right\Vert(\gamma_2)   \right) \\ & = \left\{ \begin{array}{l l}
	1 \otimes 1 & \text{if } \overline{\gamma}\cup \overline{\gamma}^\prime = \gamma \\ 0 &  \text{otherwise}
	\end{array}  \right. \\ & = \left\{ \begin{array}{l l}
	1 & \text{if } \overline{\gamma}\cup \overline{\gamma}^\prime = \gamma \\ 0 &  \text{otherwise}
	\end{array}  \right. \\ & = \left\Vert \sum_{i\in I} m_i + \sum_{j\in J} m_j^\prime \right \Vert(\gamma).
	\end{align*}

	\noindent However, if $\overline{\gamma} \cap \overline{\gamma}^\prime \not = \emptyset$, then by definition of the coalescing operator on w$_{pvm}$PCL formulas we get 
	\[ \left\Vert \left(  \sum_{i\in I} m_i \right) \uplus \left(  \sum_{j\in J} m_j^\prime \right) \right\Vert (\gamma) = 0   \]
	 for every $\gamma \in C(P)$ 
\end{proof}

\begin{proposition}\label{val_fnf}
	Let $\zeta \in PCL(D,P)$ which is in full normal form, i.e., $\zeta = \bigoplus_{i\in I}  \left( d_i \otimes \sum_{j\in J_i} m_{i,j} \right)$. Then 
	\begin{enumerate}[i.]
		\item $\ast \zeta \equiv \bigoplus_{I^\prime \subseteq I} \left(  {\val}(d_i)_{i\in I^\prime} \otimes \left( \biguplus_{i\in I^\prime} \sum_{j\in J_i} m_{i,j}  \right) \right) ,$
		\item $(\ast \zeta)\otimes \left( \biguplus_{i\in I} \sum_{j\in J_i} m_{i,j}  \right) \equiv \  {\val}(d_1, \dots, d_{ |I| }) \otimes \left( \biguplus_{i\in I} \sum_{j\in J_i} m_{i,j}  \right) . $
	\end{enumerate}
	
\end{proposition}

\begin{proof}
	\begin{enumerate}[i.]
		\item Let $\gamma \in C(P)$. Then we get 
		\begin{align*}
		\left\Vert \ast \zeta \right\Vert(\gamma) & =  \bigoplus_{ n>0} \bigoplus_{\bigcupdot_{i=1}^n \gamma_i = \gamma } {\val} \left( \left\Vert \zeta \right\Vert (\gamma_1), \dots, \left\Vert \zeta \right\Vert(\gamma_n)   \right) .
		\end{align*}

		\noindent By Lemma \ref{fullno5}, for every $i\in I$ there exists a unique $\overline{\gamma_i} \in C(P)$ such that for every $\gamma\in C(P)$ we have $\left\Vert \sum_{i\in J_i} m_{i,j} \right\Vert(\gamma) = 1 $ if $\gamma=\overline{\gamma_i}$ and $\left\Vert \sum_{i\in J_i} m_{i,j} \right\Vert(\gamma) = 0 $, otherwise. Hence, ${\val} \left( \left\Vert \zeta \right\Vert (\gamma_1), \dots, \left\Vert \zeta \right\Vert(\gamma_n)   \right) \not = 0$ when for every $i\in \{1,\dots, n\}$ there exists $j_i\in I$ such that $\gamma_i = \overline{\gamma_{j_i}}$ and, by definition of $\left\Vert \ast \zeta \right\Vert(\gamma)$, the sets $\gamma_1, \dots, \gamma_n$ are pairwise disjoint. Moreover, 
		\[{\val} \left( \left\Vert \zeta \right\Vert (\overline{\gamma_{j_1}}), \dots, \left\Vert \zeta \right\Vert(\overline{\gamma_{j_n}})   \right)= {\val}\left( d_{j_1}, \dots, d_{j_n} \right). \]

		\noindent Since ${\val} $ is a symmetric function and $D$ is idempotent, we get $\left\Vert \ast \zeta \right\Vert (\gamma) = \bigoplus_{I^{\prime \prime}\subseteq I } {\val} (d_i)_{i\in I^{ \prime \prime }}$ where for every $I^{\prime \prime } \subseteq I$ it holds $\gamma = \bigcupdot_{i\in I^{\prime \prime}} \overline{\gamma_i} $ or equivalently $\left\Vert \biguplus_{i\in I^{\prime \prime} } \sum_{j\in J_i} m_{i,j} \right\Vert $ $(\gamma) = 1. $ For every other $I^{\prime \prime \prime}$ subset of $I$ it holds $\left\Vert \biguplus_{ i\in I^{\prime \prime \prime } } \sum_{j\in J_i } m_{i,j} \right\Vert(\gamma)=0$. So, we get the following
		\[ \ast \zeta \equiv \bigoplus_{I^\prime \subseteq I} \left(  {\val}(d_i)_{i\in I^\prime} \otimes \left( \biguplus_{i\in I^\prime} \sum_{j\in J_i} m_{i,j}  \right) \right). \]

		\item Let $\gamma\in C(P)$. Then we get 
		\begin{align*}
		\left\Vert (\ast \zeta ) \otimes \left( \biguplus_{i\in I} \sum_{j\in J_i} m_{i,j}  \right)  \right\Vert (\gamma) & \equiv \left\Vert \ast \zeta \right\Vert (\gamma) \otimes \left\Vert \biguplus_{i\in I}\sum_{j\in J_i} m_{i,j} \right\Vert(\gamma).
		\end{align*}
		We can easily prove that $\left\Vert \biguplus_{i\in I}\sum_{j\in J_i} m_{i,j} \right\Vert (\gamma) = 1 $ if $\gamma= \bigcupdot_{i\in I} \overline{\gamma_i} $ and $\left\Vert \biguplus_{i\in I}\sum_{j\in J_i} m_{i,j} \right\Vert (\gamma) = 0 $ otherwise. If $\gamma= \bigcupdot_{i\in I} \overline{\gamma_i} $, then since $D$ is idempotent we get $\left\Vert \ast \zeta \right\Vert(\gamma) = {\val}\left( d_1, \dots, d_{|I|} \right).$ Hence, 
		
		\begin{align*}
		\left\Vert (\ast \zeta ) \otimes \left( \biguplus_{i\in I} \sum_{j\in J_i} m_{i,j}  \right)  \right\Vert (\gamma) & \equiv \left\{ \begin{array}{l l}
		{\val}(d_1, \dots, d_{|I|})  & \text{if } \gamma= \bigcupdot_{i\in I}\overline{\gamma_i} \\ 0 & \text{otherwise.}
		\end{array} \right. \\ & \equiv {\val} (d_1, \dots, d_{|I|}) \otimes \left\Vert \biguplus_{i\in I} \sum_{j\in J_i} m_{i,j} \right\Vert(\gamma) \\ & \equiv  \left\Vert {\val} (d_1, \dots, d_{|I|}) \otimes\left( \biguplus_{i\in I} \sum_{j\in J_i} m_{i,j}\right) \right\Vert(\gamma),
		\end{align*}
		and we are done.
	\end{enumerate}
	 
\end{proof}

\bigskip

\begin{theorem}\label{full_pv_monoid}
	Let $D$ be an associative, idempotent and $\oplus$-distributive pv-monoid, where $\otimes$ is commutative. Then, for every w$_{\text{pvm}}$PCL formula $\zeta \in PCL(D,P)$ we can effectively construct an equivalent w$_{\text{pvm}}$PCL formula $\zeta^\prime \in PCL(D,P)$ in full normal form which is unique up to the equivalence relation.  
\end{theorem}

\begin{proof}
	We prove our theorem by induction on the structure of w$_{\text{pvm}}$PCL formulas over $P$ and $D$. Let $\zeta = f$ be a PCL formula. Then, we conclude our claim by Proposition \ref{fullnoPCL}. Next let $\zeta = d$ with $d\in D$, then we have nothing to prove.

	In the sequel, assume that $\zeta_1,\zeta_2 \in PCL(D,P)$ and let 
	$\zeta_1^\prime=\bigoplus_{i_1\in I_1}\left(d_{i_1} \otimes  \right.$ $ \left. \sum_{j_1\in J_{i_1}} m_{i_1,j_1}\right)$, $\zeta_2 ^\prime=\bigoplus_{i_2\in I_2}\left(d_{i_2} \otimes \sum_{j_2\in J_{i_2}} m_{i_2,j_2}\right)$ be their equivalent full normal forms, respectively.

	To begin with, let $\zeta=\zeta_1 \oplus \zeta_2$. We consider the formula $\zeta_1^\prime \oplus \zeta_2^\prime$. If $\sum_{j_1\in J_{i
			_1}}m_{i_1, j_1}$ $ \not \equiv \sum_{j_2\in J_{i
			_2}}m_{i_2, j_2} $ for every $i_1\in I_1$ and $i_2\in I_2$, then we set $\zeta^\prime =\zeta_1^\prime \oplus \zeta_2^\prime.$ If this is not the case, we denote by $K$ the subset of $I_1\times I_2$ defined as $K:=\left\{(i_1,i_2) \mid  i_1 \in I_1, i_2 \in I_2 \textnormal{ and } \sum_{j_1\in J_{i_1}} m_{i_1,j_1}  \equiv \sum_{j_2\in J_{i_2}} m_{i_2,j_2}\right\}$. Let $K_1 = \{ i_1 \mid (i_1, i_2) \in K \}$ and $K_2 = \{ i_2 \mid (i_1, i_2) \in K \}$. Then,

	\begin{align*}
	\zeta^\prime & = \bigoplus_{i_1\in I_1\setminus K_1}\left(d_{i_1} \otimes \sum_{j_1\in J_{i_1}} m_{i_1,j_1}\right) \oplus \bigoplus_{i_2\in I_2\setminus K_2}\left(d_{i_2} \otimes \sum_{j_2\in J_{i_2}} m_{i_2,j_2}\right) \oplus \\ & \hspace*{6cm} \bigoplus_{(k_1, k_2)\in K}\left( (d_{k_1}\oplus d_{k_2}) \otimes \sum_{j\in J_{k_1}} m_{k_1,j}\right)
	\end{align*}

	\noindent  Hence, we conclude to a full normal form $\zeta^\prime $, which by construction, it is equivalent to $\zeta.$

	Next we assume that $\zeta=\zeta_1 \otimes \zeta_2$. Let the formula 
	\[\zeta^\prime= \bigoplus_{(k_1, k_2)\in K} \left((d_{k_1} \otimes d_{k_2}) \otimes \sum_{j\in J_{k_1}} m_{k_1,j}\right) \]
	where $K$ is defined as previously. We claim that $\zeta^\prime$ is the equivalent formula of $\zeta$ in full normal form. In order to prove our claim, we set \\ $ \xi = \bigoplus_{i_1\in I_1} \bigoplus_{i_2\in I_2}  \left( d_{i_1} \otimes d_{i_2} \otimes \sum_{j_1\in J_{i_1}} m_{i_1,j_1}\otimes \sum_{j_2\in J_{i_2}} m_{i_2,j_2}\right) $ and we have
	
	\begin{align*}
	\xi &  = \bigoplus_{i_1\in I_1} \bigoplus_{i_2\in I_2}  \left( d_{i_1} \otimes d_{i_2} \otimes \sum_{j_1\in J_{i_1}} m_{i_1,j_1}\otimes \sum_{j_2\in J_{i_2}} m_{i_2,j_2}\right) \\ & \equiv \bigoplus_{i_1\in I_1} \bigoplus_{i_2\in I_2}  \left(\left(d_{i_1} \otimes \sum_{j_1\in J_{i_1}} m_{i_1,j_1}\right) \otimes \left(d_{i_2} \otimes \sum_{j_2\in J_{i_2}} m_{i_2,j_2}\right)\right)  \\ & \equiv \bigoplus_{i_1\in I_1}\left(d_{i_1} \otimes \sum_{j_1\in J_{i_1}} m_{i_1,j_1}\right) \otimes \bigoplus_{i_2\in I_2}\left(d_{i_2} \otimes \sum_{j_2\in J_{i_2}} m_{i_2,j_2}\right) \\ & \equiv \zeta_1^\prime\otimes\zeta_2^\prime  \equiv \zeta_1\otimes\zeta_2  =\zeta
	\end{align*}

	\noindent where the first equivalence holds since $\otimes $ is commutative and associative and the second one since $D$ is $\oplus$-distributive.

	In the sequel, we translate $\xi$ to its equivalent full normal form $\zeta^\prime$. By Lemma \ref{fullno8}, for every $i_1\in I_1$ and $i_2\in I_2$, we have that $\sum_{j_1\in J_{i_1}} m_{i_1,j_1}\otimes \sum_{j_2\in J_{i_2}} m_{i_2,j_2}\equiv \sum_{j_1\in J_{i_1}} m_{i_1,j_1}$ if $\sum_{j_1\in J_{i_1}} m_{i_1,j_1}\equiv \sum_{j_2\in J_{i_2}} m_{i_2,j_2}$ and $\sum_{j_1\in J_{i_1}} m_{i_1,j_1}$ $\otimes $  $\sum_{j_2\in J_{i_2}} $ $ m_{i_2,j_2}\equiv 0 $ otherwise. Hence, for every $(k_1, k_2) \in K$ it holds $\sum_{j_1\in J_{k_1}} m_{k_1,j_1}$ $\otimes \sum_{j_2\in J_{k_2}} m_{k_2,j_2}\equiv \sum_{j_1\in J_{k_1}} m_{k_1,j_1}$ and for every $(i_1, i_2) \in I_1\times I_2 \backslash K$ we have that $\sum_{j_1\in J_{i_1}} m_{i_1,j_1}\otimes \sum_{j_2\in J_{i_2}} m_{i_2,j_2}\equiv 0.$ So, we conclude our claim that $\zeta^\prime $ is the required full normal form.

	Let $\zeta= \zeta_1\uplus \zeta_2$. We set \[\zeta^\prime=\bigoplus_{i_1\in I_1} \bigoplus_{i_2\in I_2}\left(d_{i_1}^\prime\otimes d_{i_2}^\prime\otimes \left(\sum_{j_1\in J_{i_1}}m_{i_1,j_1}+ \sum_{j_2\in J_{i_2}}m_{i_2,j_2}\right) \right).\] 
	 The values $d_{i_1}^\prime$ and $d_{i_2}^\prime$ are defined for every $i_1 \in I_1$ and $i_2\in I_2$ respectively, as follows. If $m_{i_1,j_1}\not \equiv m_{i_2,j_2}$ for every $ j_1\in J_1$ and $ j_2\in J_2$, then we set $d_{i_1}^\prime=d_{i_1}$ and $d_{i_2}^\prime=d_{i_2}$, otherwise we let  $d_{i_1}^\prime=d_{i_2}^\prime=0$. Then we get the following
	
	\begin{align*}
	\zeta^\prime & = \bigoplus_{i_1\in I_1} \bigoplus_{i_2\in I_2}\left(d_{i_1}^\prime\otimes d_{i_2}^\prime\otimes \left(\sum_{j_1\in J_{i_1}}m_{i_1,j_1}+ \sum_{j_2\in J_{i_2}}m_{i_2,j_2}\right) \right) \\ & \equiv \bigoplus_{i_1\in I_1} \bigoplus_{i_2\in I_2}\left(d_{i_1}\otimes d_{i_2}\otimes \left(\sum_{j_1\in J_{i_1}}m_{i_1,j_1} \uplus \sum_{j_2\in J_{i_2}}m_{i_2,j_2}\right) \right) \\ &  \equiv \bigoplus_{i_1\in I_1} \bigoplus_{i_2\in I_2}\left(  \left(d_{i_1}\otimes \sum_{j_1\in J_{i_1}} m_{i_1, j_1} \right) \uplus \left(  d_{i_2} \sum_{j_2\in J_{i_2}}m_{i_2,j_2}\right) \right) \\ &  \equiv \bigoplus_{i_1\in I_1} \left(  \left(d_{i_1}\otimes \sum_{j_1\in J_{i_1}} m_{i_1, j_1} \right) \uplus \left(\bigoplus_{i_2\in I_2}\left(  d_{i_2} \sum_{j_2\in J_{i_2}}m_{i_2,j_2}\right) \right) \right) \\ &  \equiv \left(\bigoplus_{i_1\in I_1}  \left(d_{i_1}\otimes \sum_{j_1\in J_{i_1}} m_{i_1, j_1} \right) \right)  \uplus \left( \bigoplus_{i_2\in I_2}\left(  d_{i_2} \sum_{j_2\in J_{i_2}}m_{i_2,j_2}\right) \right)  \\ & \equiv \zeta_1^\prime \uplus \zeta_2^\prime \equiv \zeta_1 \uplus \zeta_2 = \zeta
	\end{align*}
	
 \noindent where the first and second equivalences hold by Propositions \ref{w_coal_unw} and \ref{otim_coale}, respectively. The third and fourth equivalences hold since $D$ is $\oplus$-distributive.

	Finally, let $\zeta = \ast \zeta_1$. We consider the formula $\zeta^\prime = \ast \zeta_1^\prime$. By Proposition \ref{val_fnf}, $\zeta^\prime$ can be equivalently written as follows 
	\[ \zeta^\prime \equiv \bigoplus_{I_1^\prime \subseteq I_1}  \left( {\val}(d_{i_1})_{i_1\in I_1^\prime} \otimes \left(\biguplus_{i_1\in I_1^\prime } \sum_{ j_1\in J_{i_1} } m_{i_1,j_1} \right) \right).  \]
	
	\noindent We consider the sets $I_1^{(1)}, \dots ,I_1^{(k)} $ with $k \in \mathbb{N}$ to be an enumeration of all $I_1^\prime $'s such that $\biguplus_{i_1\in I_1^\prime } \sum_{ j_1\in J_{i_1} }$ $ m_{i_1,j_1} \not \equiv 0.$ Hence, by Proposition \ref{w_coal_unw}, $\biguplus_{i\in I_1^{(s)}} \sum_{j\in J_i}m_{i,j} \equiv \sum_{i\in I_1^{(s)}} \sum_{j\in J_i}m_{i,j}  $ for every $s\in \{1, \dots, k\}$. Moreover, for every $s\in \{1, \dots, k \}$ we let $d_s^\prime = {\val}(d_i)_{i\in I_1^{(s)}} $. So, 
	\[ \zeta^\prime \equiv \bigoplus_{s\in \{1, \dots, k\}} \left( d_s^\prime \otimes \left( \sum_{i\in I_1^{(s)}} \sum_{j\in J_i} m_{i,j} \right) \right) . \]

	\noindent Lastly, if $ \sum_{i\in I_1^{(s)}} \sum_{j\in J_i} m_{i,j} \not  \equiv \sum_{i\in I_1^{(s^\prime)}} \sum_{j\in J_i} m_{i,j}   $ for every $s, s^\prime \in \{1, \dots, k\}$ with $s\not = s^\prime$, then we are done. However, let $\sum_{i\in I_1^{(s)}} \sum_{j\in J_i} m_{i,j}   \equiv \sum_{i\in I_1^{(s^\prime)}} \sum_{j\in J_i} m_{i,j}$ for some $s\not =s^\prime$. Then, we replace $\left( d_s^\prime \otimes \left( \sum_{i\in I_1^{(s)}} \sum_{j\in J_i} m_{i,j} \right) \right) \oplus \left( d^\prime_{s^\prime} \otimes \left( \sum_{i\in I_1^{(s^\prime)}} \sum_{j\in J_i} m_{i,j} \right) \right)$ by its equivalent formula $(d^\prime_s \oplus d^\prime_{s^\prime}) \otimes \sum_{i\in I_1^{(s)}} \sum_{j\in J_i} m_{i,j} $ We conclude to a full normal form which by construction, it is equivalent to $\zeta.$
	
	 The uniqueness of $\zeta^\prime$, up to equivalence, is derived in a straightforward way using Statements (i) and (ii) in all the above four cases of $\zeta^\prime$. 
\end{proof}

\medskip

In the sequel, we present an example where we compute the full normal form of a w$_{\text{pvm}}$PCL formula.
\begin{example}
	Let $P$ be the set of ports and $D$ a pv-monoid which satisfies the properties of Theorem \ref*{full_pv_monoid}. We consider the {\normalfont w$_{\text{pvm}}$PCL} formula
	\[ \zeta = \left( (d_1 \otimes m_1) \uplus (d_2 \otimes (m_2 \oplus m_3)) \right) \oplus (d_3 \otimes (m_4 + m_5)) \] 
	
	\noindent where $d_1, d_2, d_3 \in D$ and $m_i$ is a full monomial over $P$ for every $i\in \{1, \dots, 5\}$. We will compute the full normal form of $\zeta^\prime = \ast \zeta.$ Firstly, we compute the full normal form of $\zeta$. 
	\begin{align*}
	\zeta & = \left( (d_1 \otimes m_1) \uplus (d_2 \otimes (m_2 \oplus m_3)) \right) \oplus (d_3 \otimes (m_4 + m_5)) \\  & \equiv \left( (d_1 \otimes m_1) \uplus ((d_2 \otimes m_2) \oplus(d_2\otimes m_3)) \right) \oplus (d_3 \otimes (m_4 + m_5)) \\ & \equiv \left(  \left(  (d_1 \otimes  m_1) \uplus (d_2 \otimes m_2) \right) \oplus  \left(  (d_1 \otimes  m_1) \uplus (d_2 \otimes m_3) \right)  \right) \oplus (d_3 \otimes (m_4 + m_5)) \\ & \equiv \left(  (d_1\otimes d_2) \otimes (m_1 +m_2) \right) \oplus \left(  (d_1\otimes d_2) \otimes (m_1 +m_3) \right) \oplus (d_3 \otimes (m_4 + m_5)).
	\end{align*}
	
	\noindent By Proposition \ref{val_fnf} we get 
	\begin{align*}
	\zeta^\prime & = \ast \zeta \\ & \equiv \ast \left( \left(  (d_1\otimes d_2) \otimes (m_1 +m_2) \right) \oplus \left(  (d_1\otimes d_2) \otimes (m_1 +m_3) \right) \oplus (d_3 \otimes (m_4 + m_5)) \right) \\ & \equiv \left( {\val}(d_1 \otimes d_2) \otimes (m_1 + m_2)  \right) \oplus \left( {\val}(d_1 \otimes d_2) \otimes (m_1 + m_3)  \right) \oplus \left( {\val}(d_3) \otimes (m_4 + m_5)  \right) \oplus \\ & \hspace*{0.4cm} \left(  {\val}(d_1\otimes d_2, d_1\otimes d_2) \otimes \left( (m_1+m_2) \uplus (m_1+m_3) \right) \right) \oplus \\ & \hspace*{0.4cm} \left(  {\val}(d_1\otimes d_2,d_3) \otimes ((m_1+m_2) \uplus (m_4 + m_5) \right) \oplus \left( {\val} (d_1\otimes d_2, d_3)\otimes \left(  (m_1 + m_3) \uplus (m_4 + m_5) \right)  \right) \oplus \\ & \hspace*{0.4cm} \left( {\val}(d_1\otimes d_2, d_1\otimes d_2, d_3) \otimes \left(  (m_1 + m_2) \uplus (m_1+m_3) \uplus (m_4 + m_5)   \right) \right)  \\ & \equiv \left( (d_1 \otimes d_2) \otimes (m_1 + m_2)  \right) \oplus \left( (d_1 \otimes d_2) \otimes (m_1 + m_3)  \right) \oplus \left( d_3 \otimes (m_4 + m_5)  \right) \oplus \\ & \hspace*{0.4cm}  \left(  {\val}(d_1\otimes d_2,d_3) \otimes (m_1+m_2+m_4 + m_5) \right) \oplus \left( {\val} (d_1\otimes d_2)\otimes \left(  m_1 + m_3 + m_4 + m_5)\right)  \right) 
	\end{align*}
	which is in full normal form.

\end{example}

Next we show that the equivalence problem for w$_{\text{pvm}}$PCL formulas is decidable. For this, we will use a corresponding result from \cite{Pa:Wei}. 

\begin{theorem}\label{equa_decid}
Let $D$ be an associative, idempotent and $\oplus$-distributive pv-monoid, where $\otimes$ is commutative. Consider also a set of ports $P$. Then for every $\zeta,\xi \in PCL(D,P)$ the equality $\|\zeta\|=\|\xi\|$ is decidable. 
\end{theorem}

\begin{proof}
	By Theorem \ref{full_pv_monoid} we can effectively construct w$_{\text{pvm}}$PCL formulas $\zeta',\xi'$ in full normal form such that $\|\zeta\|=\|\zeta'\|$ and $\|\xi\|=\|\xi'\|$. Let us assume that 
	$\zeta'=\bigoplus_{i\in I}\left(d_{i}\otimes\sum_{j\in J_{i}}m_{i,j}\right)  $  and 
	$\xi'=\bigoplus_{l\in L}\left(d'_{l}\otimes\sum_{r\in M_l}m'_{l,r}\right)  $ which moreover satisfy Statements (i) and (ii). Then, by Statement (ii) we get that $\|\zeta'\|=\|\xi'\|$ iff the following requirements (1)-(3) hold:
	\begin{itemize}
		\item[1)] $\mathrm{card}(I)=\mathrm{card}(L)$,
		\item[2)] $\{d_i \mid i \in I  \}= \{d'_l \mid l\in L \}$, and
		\item[3)] 
		\begin{itemize}
			\item[a)] if $\mathrm{card}(I)=  \mathrm{card}(\{d_i \mid i \in I  \})$, then $\sum_{j\in J_{i}}m_{i,j} \equiv \sum_{r\in M_l}m'_{l,r}$ for every $i \in I$ and $l \in L$ such that $d_i=d'_l$, 
			\item[or]
			\item[b)] if $\mathrm{card}(I)>\mathrm{card}(\{d_i \mid i \in I  \})$, then we get \\ 
			$\zeta' \equiv \bigoplus_{i'\in I'}\left(d_{i'}\otimes\bigsqcup_{i\in R_{i'}}\sum_{j\in J_{i}}m_{i,j}\right)  $ where $I'  \varsubsetneq I$,  $d_{i'}$'s ($i' \in I'$) are pairwise disjoint, and $R_{i'}$ ($i' \in I'$) is the set of all $i$ in $I$ such that $d_i=d_{i'}$. Similarly, we get
			$\xi' \equiv \bigoplus_{l'\in L'}\left(d'_{l'}\otimes\bigsqcup_{l\in S_{l'}}\sum_{r\in M_l}m'_{l,r}\right)  $ 
			where $L'  \varsubsetneq L$, $d'_{l'}$'s ($l' \in L'$) are pairwise disjoint, and $S_{l'}$ ($l' \in L'$) is the set of all $l$ in $L$ such that $d'_{l}=d'_{l'}$. Then $\bigsqcup_{i\in R_{i'}}\sum_{j\in J_{i}}m_{i,j} \equiv \bigsqcup_{l\in S_{l'}}\sum_{r\in M_l}m'_{l,r}$ for every $i' \in I'$ and $l' \in L'$ such that $d_{i'}=d'_{l'}$.
		\end{itemize}
	\end{itemize}
	By Lemma \ref{fullno5} the decidability of equivalences in (3a) is reduced to decidabilty of equality of sets of interactions corresponding to full monomials, whereas the decidabilitty of equivalences in (3b) is reduced to the decidability of equality of sets whose elements are sets of interactions corresponding to full monomials. 
\end{proof}

\section{Examples}
In this section, we provide w$_{\text{pvm}}$PCL formulas which describe well-known architectures equipped with quantitative features. But first, we introduce a new symbol which we use in order to simplify the form of the formulas in our examples.   

Let $\zeta$ be a w$_{\text{pvm}}$PCL formula. By Theorem \ref{full_pv_monoid}, $\zeta $ can be written in full normal form, hence $\zeta \equiv \bigoplus_{i\in I} \left( d_i\otimes \sum_{j\in J_i}m_{i,j} \right)$. We define the \textit{full valuation} $\circledast \zeta $ of $\zeta$ by:
\begin{itemize}
	\item[-] $\circledast \zeta : = \left( \ast \zeta  \right) \otimes \left( \biguplus_{i\in I}\sum_{j\in J_i} m_{i,j} \right).$
\end{itemize}

\noindent Then, by Proposition \ref{val_fnf} we get $\circledast \zeta \equiv  {\val}(d_1, \dots, d_{ |I| }) \otimes \left( \biguplus_{i\in I} \sum_{j\in J_i} m_{i,j}  \right).$

\begin{figure}[t!]
	\begin{center}
		\begin{tikzpicture}[scale = 0.8]
		\draw  (0,0) rectangle  (1,1);
		\draw  (1.7,0) rectangle  (2.7,1) ;
		\draw  (4,0) rectangle  (5,1) ;
		\draw  (5.7,0) rectangle  (6.7,1) ;
		\draw  (8,0) rectangle  (9,1) ;
		\draw  (9.7,0) rectangle  (10.7,1) ;
		\draw  (12,0) rectangle  (13,1) ;
		\draw  (13.7,0) rectangle  (14.7,1) ;
		
		\node at (0.5,0.7) {\small $M_1$};
		\node at (2.2,0.7) {\small $M_2$};
		\node at (4.5,0.7) {\small $M_1$};
		\node at (6.2,0.7) {\small $M_2$};
		\node at (8.5,0.7) {\small $M_1$};
		\node at (10.2,0.7) {\small $M_2$};
		\node at (12.5,0.7) {\small $M_1$};
		\node at (14.2,0.7) {\small $M_2$};
		
		\draw  (0.25,0) rectangle  (0.7,0.3);
		\draw  (1.95,0) rectangle  (2.4,0.3);
		\draw  (4.25,0) rectangle  (4.7,0.3);
		\draw  (5.95,0) rectangle  (6.4,0.3);
		\draw  (8.25,0) rectangle  (8.7,0.3);
		\draw  (9.95,0) rectangle  (10.4,0.3);
		\draw  (12.2,0) rectangle  (12.7,0.3);
		\draw  (13.9,0) rectangle  (14.4,0.3);
		
		\node at (0.5,0.15) {\tiny $m_1$};
		\node at (2.2,0.15) {\tiny $m_2$};
		\node at (4.5,0.15) {\tiny $m_1$};
		\node at (6.2,0.15) {\tiny $m_2$};
		\node at (8.5,0.15) {\tiny $m_1$};
		\node at (10.2,0.15) {\tiny $m_2$};
		\node at (12.5,0.15) {\tiny $m_1$};
		\node at (14.2,0.15) {\tiny $m_2$};

		\draw  (0,-3) rectangle  (1,-2);
		\draw  (1.7,-3) rectangle  (2.7,-2) ;
		\draw  (4,-3) rectangle  (5,-2) ;
		\draw  (5.7,-3) rectangle  (6.7,-2) ;
		\draw  (8,-3) rectangle  (9,-2) ;
		\draw  (9.7,-3) rectangle  (10.7,-2) ;
		\draw  (12,-3) rectangle  (13,-2) ;
		\draw  (13.7,-3) rectangle  (14.7,-2) ;
		
		\node at (0.5,-2.8) {\small $S_1$};
		\node at (2.2,-2.8) {\small $S_2$};
		\node at (4.5,-2.8) {\small $S_1$};
		\node at (6.2,-2.8) {\small $S_2$};
		\node at (8.5,-2.8) {\small $S_1$};
		\node at (10.2,-2.8) {\small $S_2$};
		\node at (12.5,-2.8) {\small $S_1$};
		\node at (14.2,-2.8) {\small $S_2$};
		
		\draw  (0.3,-2) rectangle  (0.65,-2.3);
		\draw  (2,-2) rectangle  (2.35,-2.3);
		\draw  (4.3,-2) rectangle  (4.65,-2.3);
		\draw  (6,-2) rectangle  (6.36,-2.3);
		\draw  (8.3,-2) rectangle  (8.65,-2.3);
		\draw  (10,-2) rectangle  (10.35,-2.3);
		\draw  (12.3,-2) rectangle  (12.65,-2.3);
		\draw  (14,-2) rectangle  (14.35,-2.3);
		
		\node at (0.5,-2.2) {\tiny $s_1$};
		\node at (2.2,-2.2) {\tiny $s_2$};
		\node at (4.5,-2.2) {\tiny $s_1$};
		\node at (6.2,-2.2) {\tiny $s_2$};
		\node at (8.5,-2.2) {\tiny $s_1$};
		\node at (10.2,-2.2) {\tiny $s_2$};
		\node at (12.5,-2.2) {\tiny $s_1$};
		\node at (14.2,-2.2) {\tiny $s_2$};
		
		\draw[fill] (0.5,0) circle [radius=1.5pt];
		\draw[fill] (2.2,0) circle [radius=1.5pt];
		\draw[fill] (4.5,0) circle [radius=1.5pt];
		\draw[fill] (6.2,0) circle [radius=1.5pt];
		\draw[fill] (8.5,0) circle [radius=1.5pt];
		\draw[fill] (10.2,0) circle [radius=1.5pt];
		\draw[fill] (12.5,0) circle [radius=1.5pt];
		\draw[fill] (14.2,0) circle [radius=1.5pt];

		\draw[fill] (0.5,-2) circle [radius=1.5pt];
		\draw[fill] (2.2,-2) circle [radius=1.5pt];
		\draw[fill] (4.5,-2) circle [radius=1.5pt];
		\draw[fill] (6.2,-2) circle [radius=1.5pt];
		\draw[fill] (8.5,-2) circle [radius=1.5pt];
		\draw[fill] (10.2,-2) circle [radius=1.5pt];
		\draw[fill] (12.5,-2) circle [radius=1.5pt];
		\draw[fill] (14.2,-2) circle [radius=1.5pt];
		
		\draw  (0.5,0)--(0.5,-2);
		\node at (0.1,-0.8) { $d_{1,1}$};
		\draw  (2.2,0)--(2.2,-2);
		\node at (1.8,-0.8) { $d_{2,2}$};
		
		\draw  (4.5,0)--(4.5,-2);
		\node at (4.1,-0.8) { $d_{1,1}$};
		\draw  (4.5,0)--(6.2,-2);
		\node at (5.6,-0.8) { $d_{2,1}$};

		\draw  (8.5,0)--(10.2,-2);
		\node at (8.7,-0.8) { $d_{2,1}$};
		\draw  (10.2,0)--(8.5,-2);
		\node at (10.1,-0.8) { $d_{1,2}$};
		
		\draw  (14.2,0)--(12.5,-2);
		\node at (13,-0.8) { $d_{1,2}$};
		\draw  (14.2,0)--(14.2,-2);
		\node at (14.6,-0.8) { $d_{2,2}$};
		\end{tikzpicture}
	\end{center}
	\caption{Weighted Master/Slave architecture.}
	\label{MS_fig}
\end{figure}
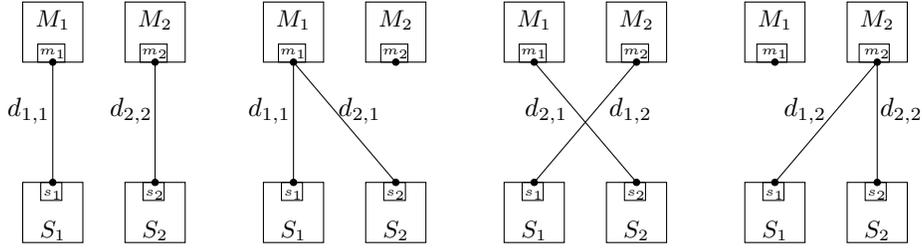

\begin{example}\label{master_slave}
	We recall from \cite{Ma:Co} the Master/Slave architecture for two masters $M_1, M_2$ and two slaves $S_1, S_2$ with ports $m_1, m_2$ and $s_1, s_2$, respectively. Masters can interact only with slaves, and vice versa, and each slave can interact with only one master. Hence, the four possible instances of the Master/Slave architecture for two masters and two slaves are shown in Figure \ref{MS_fig}. In the following we present four different {\normalfont w$_{\text{pvm}}$PCL} formulas, which according to the underlying pv-monoid we get interesting results.
	
	The monomial $\phi_{i,j} = m_{ \{ s_i, m_j \} }$ for every $i,j \in \{ 1,2 \}$ represents the binary interaction between the ports $s_i$ and $m_j$. For every $i,j\in \{ 1,2 \}$ we consider a value $d_{i,j} \in D$ and the {\normalfont w$_{\text{pvm}}$PIL} formula $\varphi_{i,j} = d_{i,j} \otimes \phi_{i,j} $. Hence, $d_{i,j} $ can be considered as the ``cost'' for the implementation of the interaction $\{ s_i,m_j \}$. For our example we consider the configuration set $\gamma=\{ \{s_1,m_1\}, \{s_1,m_2\}, \{s_2,m_1\}, \{s_2,m_2\} \}$ and the pv-monoid $(\mathbb{R}\cup\{-\infty\}, \max, $ $ {\avg}, +, $ $-\infty, 0)$.

	\begin{enumerate}[(i)]
	\item Let us assume that we want to compute the average cost of each of the possible architectures shown in Figure \ref{MS_fig} and then the maximum of those values. We consider the {\normalfont w$_{\text{pvm}}$PCL} formula
		\begin{align*}
		\zeta &= \sim \bigoplus_{i,j \in \{1,2\}} \circledast \left( \varphi_{1,i} \oplus \varphi_{2,j} \right). 
		\end{align*}
		 Then the value
		\begin{align*}
		\left\Vert  \zeta \right\Vert (\gamma)    &  =  \left\Vert \sim \bigoplus_{i,j \in \{1,2\}} \circledast \left( \varphi_{1,i} \oplus \varphi_{2,j} \right) \right\Vert (\gamma) \\ & = \max_{\gamma^\prime \subseteq \gamma} \left\{  \left\Vert \bigoplus_{i,j \in \{1,2\}} \circledast \left( \varphi_{1,i} \oplus \varphi_{2,j} \right) \right\Vert (\gamma^\prime)  \right\} \\ & = \max_{\gamma^\prime \subseteq \gamma} \left\{ \max_{i,j\in \{1,2\}} \left\{\left\Vert \circledast \left( (d_{1,i} \otimes m_{\{ s_1, m_i \}}) \oplus (d_{2,j} \otimes m_{ \{s_2, m_j\} } \right) \right\Vert (\gamma^\prime)  \right\} \right\}  \\ & = \max_{\gamma^\prime \subseteq \gamma} \left\{ \max_{i,j\in \{1,2\}} \left\{ {\avg} (d_{1,i}, d_{2,j}) + \left\Vert m_{ \{ s_1, m_i \}  } \uplus m_{ \{ s_2, m_j \} } \right\Vert (\gamma^\prime) \right\} \right\} \\ & = \max \left\{  {\avg} (d_{1,1}, d_{2,1}) , {\avg} (d_{1,1}, d_{2,2}) , {\avg} (d_{1,2}, d_{2,1}) , {\avg} (d_{1,2}, d_{2,2}) \right\}
		\end{align*}
		
		\noindent computes the average cost for each of the four possible instances and then the maximum of those values. It is interesting to note that $\left\Vert \zeta \right\Vert (\gamma) = \left\Vert \zeta \right\Vert (\gamma^\prime)$ for every $\gamma^\prime \in C(P)$ with $\gamma \subseteq \gamma^\prime. $

		\medskip
		
		\item Moreover, let the following {\normalfont w$_{\text{pvm}}$PCL} formula 
		\begin{align*}
		\zeta & =  \bigotimes_{i,j \in \{1,2\}} \sim \left( \circledast \left( \varphi_{1,i} \oplus \varphi_{2,j} \right) \right).
		\end{align*}
		
		\noindent
      Then, the value 
		\begin{align*}
		\left\Vert \zeta \right\Vert (\gamma) & = \sum_{i,j\in \{ 1,2 \}} \left( {\max}_{\gamma^\prime \subseteq \gamma} \left\{  {\avg} (d_{1,i}, d_{2,j}) + \left\Vert m_{  \{ s_1, m_i \} } \uplus m_{ \{ s_2, m_j \} }  \right\Vert (\gamma^\prime) \right\} \right) \\ & = {\avg}(d_{1,1}, d_{2,1}) + {\avg}(d_{1,2}, d_{2,1}) + {\avg}(d_{1,1}, d_{2,2}) + {\avg}(d_{1,2}, d_{2,2})
		\end{align*}
		 is the sum of the average costs of all architecture schemes.
		
		\item As a third case, we want to compute the slave which has the maximum average cost with the existing masters. Therefore, we consider the following {\normalfont w$_{\text{pvm}}$PCL} formula:
		\begin{align*}
		\zeta &= \sim \bigoplus_{ i\in \{1,2\} } \left(  \circledast  \left( \varphi_{i,1} \oplus \varphi_{i,2} \right) \right).
		\end{align*}
		
		\noindent
	 Then we get 	
		\begin{align*}
		\left\Vert\zeta \right\Vert(\gamma) & =  \left\Vert \sim \bigoplus_{ i\in \{1,2\} } \left(  \circledast  \left( \varphi_{i,1} \oplus \varphi_{i,2} \right) \right) \right\Vert (\gamma) \\ & = \max_{\gamma^\prime \subseteq \gamma} \left\{ \max_{i\in  \{1,2\}}  \left\{  {\avg} (d_{i,1}, d_{i2}) + \left\Vert m_{ \{s_i, m_1\} }  + m_{ \{s_i, m_2\} } \right\Vert(\gamma^\prime) \right\}   \right\} \\ & = \max  \{ {\avg} (d_{1,1}, d_{1,2}) , {\avg} (d_{2,1}, d_{2,2})  \}
		\end{align*}
		which is the wanted outcome.

	\end{enumerate}
\end{example}

\bigskip

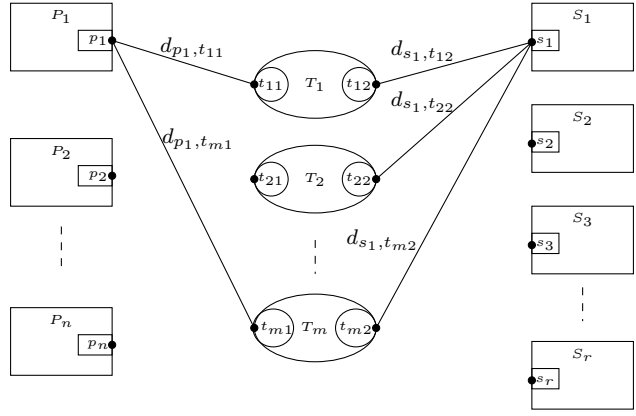
\begin{figure}[t!]
	\begin{center}
		\begin{tikzpicture}[scale=0.9]
		
		\draw  (-0.5,0) rectangle  (1,1) ;
		\node at (0.25,0.8) {\tiny $P_1$};
		\draw  (0.5,0.3) rectangle (1,0.6);
		\node at (0.8,0.45) {\tiny $p_1$};
		\draw[fill] (1,0.45) circle [radius=1.5pt];

		\draw  (-0.5,-2) rectangle  (1,-1) ;
		\node at (0.25,-1.2) {\tiny $P_2$};
		\draw  (0.5,-1.7) rectangle (1,-1.4);
		\node at (0.8,-1.55) {\tiny $p_2$};
		\draw[fill] (1,-1.55) circle [radius=1.5pt];

		\draw[dashed] (0.25,-2.3)--(0.25,-3);
		
		\draw  (-0.5,-3.5) rectangle  (1,-4.5) ;
		\node at (0.25,-3.7) {\tiny $P_n$};
		\draw  (0.5,-4.2) rectangle (1,-3.9);
		\node at (0.8,-4.05) {\tiny $p_n$};
		\draw[fill] (1,-4.05) circle [radius=1.5pt];

		\draw (4,-0.2) ellipse (0.9cm and 0.5cm);
		\node at (4,-0.2) {\tiny $T_1$};
		\draw (3.35,-0.2) ellipse (0.25cm and 0.25cm);
		\node at (3.35,-0.2) {\tiny $t_{11}$};
		\draw (4.65,-0.2) ellipse (0.25cm and 0.25cm);
		\node at (4.65,-0.2) {\tiny $t_{12}$};
		\draw[fill] (3.1,-0.2) circle [radius=1.5pt];
		\draw[fill] (4.9,-0.2) circle [radius=1.5pt];

		\draw (4,-1.6) ellipse (0.9cm and 0.5cm);
		\node at (4,-1.6) {\tiny $T_2$};
		\draw (3.35,-1.6) ellipse (0.25cm and 0.25cm);
		\node at (3.35,-1.6) {\tiny $t_{21}$};
		\draw (4.65,-1.6) ellipse (0.25cm and 0.25cm);
		\node at (4.65,-1.6) {\tiny $t_{22}$};
		\draw[fill] (3.1,-1.6) circle [radius=1.5pt];
		\draw[fill] (4.9,-1.6) circle [radius=1.5pt];

		\draw[dashed] (4,-2.5)--(4,-3);
		
		\draw (4,-3.8) ellipse (0.9cm and 0.5cm);
		\node at (4,-3.8) {\tiny $T_m$};
		\draw (3.38,-3.8) ellipse (0.3cm and 0.28cm);
		\node at (3.35,-3.8) {\tiny $\ \ t_{m1} $};
		\draw (4.6,-3.8) ellipse (0.3cm and 0.28cm);
		\node at (4.65,-3.8) {\tiny $t_{m2} \ $};
		\draw[fill] (3.1,-3.8) circle [radius=1.5pt];
		\draw[fill] (4.9,-3.8) circle [radius=1.5pt];

		\draw  (7.2,0) rectangle  (8.7,1) ;
		\node at (7.95,0.8) {\tiny $S_1$};
		\draw  (7.2,0.3) rectangle (7.6,0.6);
		\node at (7.4,0.425) {\tiny $s_1$};
		\draw[fill] (7.2,0.425) circle [radius=1.5pt];
		
		\draw  (7.2,-1.5) rectangle  (8.7,-0.5) ;
		\node at (7.95,-0.7) {\tiny $S_2$};
		\draw  (7.2,-1.2) rectangle (7.6,-0.9);
		\node at (7.4,-1.075) {\tiny $s_2$};
		\draw[fill] (7.2,-1.075) circle [radius=1.5pt];
		
		\draw  (7.2,-3) rectangle  (8.7,-2) ;
		\node at (7.95,-2.2) {\tiny $S_3$};
		\draw  (7.2,-2.7) rectangle (7.6,-2.4);
		\node at (7.4,-2.575) {\tiny $s_3$};
		\draw[fill] (7.2,-2.575) circle [radius=1.5pt];
		
		\draw[dashed] (7.95,-3.2)--(7.95,-3.7);
		
		\draw  (7.2,-5) rectangle  (8.7,-4) ;
		\node at (7.95,-4.2) {\tiny $S_r$};
		\draw  (7.2,-4.7) rectangle (7.6,-4.4);
		\node at (7.4,-4.575) {\tiny $s_r$};
		\draw[fill] (7.2,-4.575) circle [radius=1.5pt];

		\draw (1,0.45)--(3.1,-0.2);
		\node at (2.2,0.35) { \footnotesize $d_{p_1, t_{11}}$};
		\draw  (1,0.45)--(3.1,-3.8);
		\node at (2.28,-1) { \footnotesize $d_{p_1, t_{m1}}$};
		
		\draw (4.9,-0.2)--(7.2,0.425);
		\node at (5.6,0.3) { \footnotesize $d_{s_1, t_{12}}$};

		\draw (4.9,-1.6)--(7.2,0.425);
		\node at (5.6,-0.45) { \footnotesize $d_{s_1,t_{22}}$};

		\draw (4.9,-3.8)--(7.2,0.425);			
		\node at (5,-2.5) { \footnotesize $d_{s_1, t_{m2}}$};

		\end{tikzpicture}  
		
	\end{center}
	\caption{Weighted Publish/Subscribe architecture.}
	\label{pub-sub}
\end{figure}

\begin{example}
	Publish/Subscribe is a software architecture with three types of components namely, publishers, topics, and subscribers denoted by the letters $P,T,S$, respectively (cf. \cite{Eu:Th, Ha:Ap, Pa:On}). Publishers send messages to subscribers but they do not have any information about subscribers and vice versa. So, in order to send messages, publishers characterize messages according to classes/topics. Subscribers, on the other hand, express their interest in one or more topics and receive all messages which have been published to the topics to which they subscribe (Figure \ref{pub-sub}).

	In our example we assign weights, describing priorities, to interactions among publishers and topics, and to interactions among topics and subscribers. Component $P$ has one port $p$, $T$ has two ports
	$t_{1}$ and $t_{2}$, and $S$ has the port $s$.  We assume two publisher components $P_1, P_2$, four subscriber components $S_1, S_2, S_3, S_4$ and three topic components $T_1, T_2, T_3$. Hence, the set of ports is $P=\left\{ p_1,p_2, s_1, s_2, s_3, s_4, t_{11}, t_{12}, t_{21}, t_{22}, t_{31}, t_{32} \right\}$. For every $i\in \{1,2,3,4\}$, $j\in \{1,2,3\}$ and $k\in \{ 1,2 \}$ we denote by $d_{s_i,t_{j2}} \in D$ the weight of the interaction among $S_i$ and $T_j$, i.e., the priority that the subscriber $S_i$ assigns to the receivement of a message from $T_j$, and by $d_{p_k,t_{j1}} \in D$, the weight of the interaction among $P_k$ and $T_j$, i.e., the priority that the topic $T_j$ assigns to the receivement of a message from $P_k.$ 
	
	In the sequel, we develop {\normalfont w$_{\text{pvm}}$PCL} formulas whose semantics compute the maximum average priority with which a subscriber will receive a message and also the maximum most frequent priority of each topic. For every $i\in \{1,2\}$ and $j\in \{ 1,2,3 \}$, the {\normalfont w$_{\text{pvm}}$PIL} formula $\varphi_{pt}(p_i, t_{j1}) = d_{ p_i,t_{j1} } \otimes m_{ \{ p_i, t_{j1} \} }$ characterizes the interaction between a publisher $P_i$ and a topic $T_j$ with its corresponding weight. Moreover, for every $i\in \{ 1,2,3,4 \}$ and $j\in \{ 1,2,3 \}$, the {\normalfont w$_{\text{pvm}}$PIL} $\varphi_{ st }( s_i,t_{j2} ) = d_{s_i,t_{j2}} \otimes m_{ \{ s_i,t_{j2} \} } $ characterizes the interaction between a subscriber $S_i$ and a topic $T_j$ with its corresponding weight. Then, the {\normalfont w$_{\text{pvm}}$PCL} formula
	\[   \zeta_{s_i} = \bigoplus_{ j\in \{1,2,3\} } \bigoplus_{k\in \{ 1,2 \}}\circledast \left(   \varphi_{pt}(p_k, t_{j1}) \oplus  \varphi_{st}(s_i,t_{j2}) \right)  \]
	
	\noindent describes the behavior of subscriber $S_i$ with publishers $P_1, P_2$ and topics $T_1, T_2, T_3$. Let the configuration set 
	\begin{align*}
	\gamma = & \left\{ \{p_1, t_{11}\} , \{p_1, t_{21}\} , \{p_1, t_{31}\}, \{p_2, t_{11}\} , \{p_2, t_{21}\} , \{p_2, t_{31}\}, \{s_1, t_{12} \}, \{s_1, t_{22} \}, \right. \\ &  \left.  \{s_1, t_{32} \},   \{s_2, t_{12} \}, \{s_2, t_{22} \}, \{s_2, t_{32} \}, \{s_3, t_{12} \}, \{s_3, t_{22} \}, \{s_3, t_{32} \}  \right\}, 
	\end{align*}
	and the pv-monoid $(\mathbb{R}\cup\{-\infty\}, \max, $ $ {\avg}, +, -\infty, 0)$. Then the value $\left\Vert \sim \zeta_{s_i}\right\Vert (\gamma) $ represents the maximum average priority with which the subscriber $S_i$ will receive a message. Also, consider the w$_{\text{pvm}}$PCL formula $\zeta = \bigotimes_{i\in \{ 1,2,3,4 \}}\left( \sim \zeta_{s_i}  \right)$. Then, the following value \[\left\Vert \zeta \right\Vert (\gamma) = \sum_{ i \in \{1,2,3,4\}} \left(\max_{j\in \{ 1,2,3 \}} \left\{  {\avg}(d_{p_1, t_{j1}}, d_{s_i, t_{j2}}),  {\avg}(d_{p_2, t_{j1}}, d_{s_i, t_{j2}})  \right\}  \right) \]

	\noindent is the sum of the values $\left\Vert \sim \zeta_{s_i}\right\Vert (\gamma)$ for $ i \in \{ 1,2,3,4 \} $.
	
	Moreover, let us assume that we want to erase one component of the architecture in case, for example, where the system is overloaded and needs to be `lightened'. Consider the case where we choose to erase a topic which is not as popular as the others. A way to do this is to compute for every topic the most frequent priorities that the publishers and subscribers give to that component and then the maximum one of those. Hence, the topic that has the minimum most frequent priority among the other topics is the least popular topic and so it can be erased. The following {\normalfont w$_{\text{pvm}}$PCL} formula
	\[  \zeta_{t_i} = \circledast \left(\bigoplus_{j\in \{ 1,2 \}} \varphi_{pt}(p_j, t_{i1}) \oplus  \bigoplus_{k\in \{ 1,2,3,4 \}} \varphi_{st}(s_k,t_{i2}) \right)  \] 
	for $i\in \{ 1,2,3 \}$ describes the full valuation of the weighted interactions of the topic $T_i$ with the publishers $P_1$, $P_2$ and the subscribers $S_1, S_2, S_3$ and $S_4$. Consider the configuration $\gamma$ given above and the pv-monoid $(\mathbb{R}\cup \{ +\infty, -\infty \},$ $ \min, {\maj}, \max, +\infty, -\infty)$. Then, 
	\begin{align*}
	 \left\Vert \sim \zeta_{t_i} \right\Vert (\gamma) & = \min_{\gamma^\prime \subseteq \gamma}\left\{ \max \left\{ {\maj} \left( d_{p_1, t_{i1}}, d_{p_2, t_{i1}}, d_{s_1, t_{i2}}, d_{s_2, t_{i2}},d_{s_3, t_{i2}}, d_{s_4, t_{i2}}\right), \right. \right. \\ & \left. \left. \hspace*{4cm}  \left\Vert \biguplus_{j\in \{ 1,2 \}} \phi_{pt}(p_j, t_{i1}) \uplus  \biguplus_{k\in \{ 1,2,3,4 \}} \phi_{st}(s_k,t_{i2}) \right\Vert (\gamma^\prime)  \right\}   \right\} \\ & = {\maj} \left( d_{p_1, t_{i1}}, d_{p_2, t_{i1}}, d_{s_1, t_{i2}}, d_{s_2, t_{i2}},d_{s_3, t_{i2}}, d_{s_4, t_{i2}}\right)
	\end{align*}
	
	\noindent for $i\in \{ 1,2,3 \}$ is the maximum priority, among the most frequent ones, that the publishers and subscribers give to topic $T_i$. Lastly, if we consider the w$_{\text{pvm}}$PCL formula
	\[  \zeta^\prime = \sim \left(\zeta_{t_1} \oplus \zeta_{t_2} \oplus \zeta_{t_3} \right),  \]
	 then $\left\Vert \zeta^\prime \right\Vert (\gamma) = \min_{i\in  \{1,2,3\} } \left\{  {\maj} \left( d_{p_1, t_{i1}}, d_{p_2, t_{i1}}, d_{s_1, t_{i2}}, d_{s_2, t_{i2}},d_{s_3, t_{i2}}, d_{s_4, t_{i2}}\right) \right\}$ and so we erase the topic with the minimum value.

\end{example}

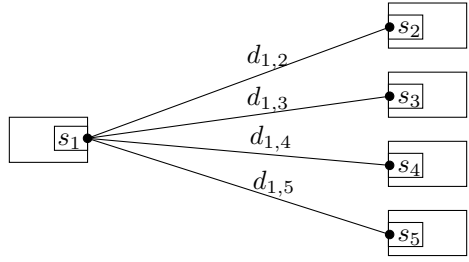
\begin{figure}[t!]
	\begin{center}
		\begin{tikzpicture}[scale = 0.8]
		\draw  (0.7,-9.75) rectangle  (2,-9); 
		\draw  (7,-7.85) rectangle  (8.3,-7.1); 
		\draw  (7,-9) rectangle  (8.3,-8.25);
		\draw  (7,-10.15) rectangle  (8.3,-9.4);
		\draw  (7,-11.3) rectangle  (8.3,-10.55);

		\draw  (1.45,-9.55) rectangle  (2,-9.15); 
		\node at (1.7,-9.4) {$ s_1$};
		\draw[fill] (2,-9.35) circle [radius=2pt];
		
		\draw  (7,-7.7) rectangle  (7.57,-7.3); 
		\node at (7.35,-7.54) {$ s_2$};
		\draw[fill] (7.02,-7.5) circle [radius=2pt];
		
		\draw  (7,-8.85) rectangle  (7.57,-8.45); 
		\node at (7.35,-8.7) {$ s_3$};
		\draw[fill] (7.02,-8.65) circle [radius=2pt];
		
		\draw  (7,-10) rectangle  (7.57,-9.6); 
		\node at (7.35,-9.85) {$ s_4$};
		\draw[fill] (7.02,-9.8) circle [radius=2pt];
		
		\draw  (7,-11.15) rectangle  (7.57,-10.75); 
		\node at (7.35,-11) {$ s_5$};
		\draw[fill] (7.02,-10.95) circle [radius=2pt];
		
		\draw (2,-9.35)--(7.02,-7.5);
		\draw (2,-9.35)--(7.02,-8.65);
		\draw (2,-9.35)--(7.02,-9.8);
		\draw (2,-9.35)--(7.02,-10.95);
		
		\node at (5,-8.0) {$ d_{1,2}$};
		\node at (5,-8.7) {$ d_{1,3}$};
		\node at (5.05,-9.35) {$ d_{1,4}$};
		\node at (5.1,-10.1) {$ d_{1,5}$};

		\end{tikzpicture}
	\end{center}
	\caption{Weighted Star architecture.}
	\label{star_examp}
\end{figure}

\begin{example}
	Consider the Star architecture \cite{Ma:Co}. Star architecture is a software architecture relating components of the same type.
 Given a set of components one of them is considered as the central one and is connected to every other component through a binary interaction. No other interactions are permitted. 
	
	In our example we consider five components (Figure \ref{star_examp}). We assume that each component has a single port, hence the set of ports is $P=\{ s_1, s_2, s_3, s_4, s_5 \}$. We denote by $d_{i,j} \in D$ the weight of the binary interaction between $s_i$ and $s_j$ for every $i,j\in I=\{1, \dots, 5\}$ with $i\not = j$, when $s_i$ is considered as the central component. The {\normalfont w$_{\text{pvm}}$PIL} formula characterizing this interaction, for every $i,j\in I$ with $i\not = j$, is given by $\varphi_{ij} = d_{i,j} \otimes m_{\{s_i,s_j\}  }$. Therefore, the {\normalfont w$_{\text{pvm}}$PCL} formula 
	\[\zeta_i=\circledast \left(\bigoplus_{j\in I\backslash \{i\}}\varphi_{ij}\right)\]
	describes the full valuation of the binary interactions of the central component $s_i$ with the rest of all other components. Next, consider the {\normalfont w$_{\text{pvm}}$PCL} formula 
	\[\zeta = \sim \left(\bigoplus_{ i \in I } \zeta_i\right)\]
	which describes the five alternative versions of the Star architecture. Let $\gamma = \{ \{s_i,s_j\}/\ i,j \in I $ and  $i\neq j\}$ and $(\mathbb{R}  $  $ \cup $ $\{+\infty\}, \min, {\avg}, +, +\infty, 0)$, then we get 
	\begin{align*}
	\left\Vert \zeta \right\Vert (\gamma) & = \min_{\gamma^\prime \subseteq \gamma}  \left\{  \min_{i\in I} \{ \left\Vert \zeta_i\right\Vert(\gamma^\prime)  \} \right\} \\ & =\min\{ {\avg} (d_{1,2},d_{1,3},d_{1,4},d_{1,5}),...,{\avg} (d_{5,1},d_{5,2},d_{5,3},d_{5,4})\} 
	\end{align*}

  \noindent	which is the minimum value among the average costs of each component when it is considered as the central one.  
	
\end{example}

\section{Conclusion}

We introduced a weighted PCL over a set of ports and a pv-monoid, and investigated several properties of the class of polynomials obtained as semantics of this logic with the condition that our pv-monoid satisfies specific properties. We proved that for every w${_{\text{pvm}}}$PCL formula $\zeta $ over a set of ports $P$ and a pv-monoid $D$ which is associative, $\oplus$-distributive, idempotent and $\otimes $ is commutative, we can effectively construct an equivalent one $\zeta^\prime$ in full normal form. This result implied the decidability of the equivalence problem for w$_{\text{pvm}}$PCL formulas. Lastly, we provided examples describing well-known software architectures with quantitative characteristics such as the average cost of an architecture or the maximum most frequent priority of a component in the architecture. These are important properties which can not be represented by the framework of semirings in \cite{Pa:On}. Future work includes the investigation of the complexity for the construction of full normal form for formulas in our logic and the time needed for that construction using the Maude rewriting system \cite{maude}. Furthermore, it would be  interesting to study the first-order level of {\normalfont w$_{\text{pvm}}$PCL } for the description of architecture styles with quantitative features.

\end{document}